\documentclass[11pt,draftcls,onecolumn]{IEEEtran}
\usepackage{amsmath}
\usepackage{amsthm}
\usepackage{graphicx}
\usepackage{color}
\usepackage{subfigure}
\usepackage{cite}
\usepackage{url}

\usepackage{algorithm}
\usepackage{algorithmic}

\theoremstyle{plain}
\newtheorem{thm}{Theorem}

\newtheorem{lem}[thm]{Lemma}

\theoremstyle{definition}
\newtheorem{dfn}{Definition}

\title{Chinese Restaurant Game - Part I: Theory of Learning with Negative Network Externality}
\author{
\IEEEauthorblockN{Chih-Yu Wang$^{1,2}$, Yan Chen$^1$, and K. J. Ray Liu$^1$}\\
\IEEEauthorblockA{
$^1$Department of Electrical and Computer Engineering, University of Maryland, College Park, MD 20742 USA\\
$^2$Graduate Institute of Communication Engineering, National Taiwan University, \\
Taipei, Taiwan
}
}
\begin{document}

\maketitle

\begin{abstract}
In a social network, agents are intelligent and have the capability to make decisions to maximize their utilities. They can either make wise decisions by taking advantages of other agents' experiences through learning, or make decisions earlier to avoid competitions from huge crowds.
Both these two effects, social learning and negative network externality, play important roles in the decision process of an agent. While there are existing works on either social learning or negative network externality, a general study on considering both these two contradictory effects is still limited. We find that the Chinese restaurant process, a popular random process, provides a well-defined structure to model the decision process of an agent under these two effects. By introducing the strategic behavior into the non-strategic Chinese restaurant process, in Part I of this two-part paper, we propose a new game, called Chinese Restaurant Game, to formulate the social learning problem with negative network externality.
Through analyzing the proposed Chinese restaurant game, we derive the optimal strategy of each agent and provide a recursive method to achieve the optimal strategy.
How social learning and negative network externality influence each other under various settings is also studied through simulations.
\end{abstract}

\newpage

\section{Introduction}
How agents in a network learn and make decisions is an important issue in numerous research fields, such as social learning in social networks, machine learning with communications among devices, and cognitive adaptation in cognitive radio networks. Agents make decisions in a network in order to achieve certain objectives. For example, one customer goes to the supermarket for an orange juice. He may need to choose one from dozens of brands. However, the agent's knowledge on the market may be very limited due to the limited ability in observations or the external uncertainty in the market, which means that the customer may not know the quality of all orange juice in different brands. This limitation reduces the accuracy of the agent's decision for his objective, e.g., to get the best orange juice of his taste.

The limited knowledge of one agent can be expanded through learning. One agent may learn from some information sources, such as the decisions of other agents, the advertisements from some brands, or his experience in previous purchases. All the information can help the agent to construct a belief, which is mostly probabilistic, on the unknown state. In most cases, the accuracy of the agent's decision can be greatly enhanced by taking into account the belief. A general learning and decision making process in a network can be described as follows. First, an agent collects information through available communication or observation methods and updates his belief on the uncertain states based on the collected information. Then, the agent estimates the expected rewards of certain actions according to the belief he constructed. Finally, the agent chooses the action that maximizes his reward.

Let us consider a social network in an uncertain system state. The state has an impact on the agents' rewards. When the impact is differential, i.e., one action results in a higher reward than other actions in one state but not in all states, the state information becomes critical for one agent to make the correct decision. In most social learning literatures, the state information is unknown to agents. Nevertheless, some signals related to the system state are revealed to the agents. These signals may be preserved in private or revealed to others. Then, the agents make their decisions sequentially, while their actions/signals may be fully or partially observed by other agents. Most of existing works \cite{bala1998learning,golub2007naive,acemoglu2011bayesian,acemoglu2010opinion} study how the believes of agents are formed through learning in the sequential decision process, and how accurate the believes will be when more information is revealed. One popular assumption in traditional social learning literatures is that there is no network externality, i.e., the actions of subsequent agents do not influent the reward of the former agents. In such a case, agents will make their decisions purely based on their own believes without considering the actions of subsequent agents. This assumption greatly limits the potential applications of these existing works.

The network externality, i.e., the influence of other agents' behaviors on one agent's reward, is a classic topic in economics. How the relations of agents influence an agent's behavior is one of the major problems in coordinate game theory \cite{cooper1999coordination}. When the network externality is positive, the problem can be modeled as a coordination game, where agents seek the best common decisions to cooperate with others. When the externality is negative, it becomes an anti-coordination game, where agents try to avoid making the same decisions with others \cite{katz1986technology,sandholm2005negative,fagiolo2005endogenous}.

In the literature, there are some works on combining the positive network externality with social learning, such as voting game \cite{wit1999social,battaglini2005sequential,ali2010observ} and investment game \cite{gale1995dynamic,dasgupta2000social,dasgupta2007coordination,choi2011network}. In the voting game, an election with several candidates is hold, where voters have their own preferences on the candidates. The preference of a voter on the candidates is constructed by the voter's belief on how the candidates can benefit him if winning the election. However, since the candidate can make efforts only when he wins the election, a voter's vote depends not only on his own preference but also on the probability that the candidate wins the election. In such a case, the estimation and prediction on the decisions of other voters become critical in the voting game. A learning process is involved when the voting game is sequential, i.e., voters vote the candidates sequentially and the vote of each voter is known by others. In the sequential voting game, voters learn from the previous votes to update their believes on the candidates and the probability that the candidates win the election.

In the investment game, there are multiple projects and investors, where each project has different probability of success and different payoff. One investor may invest one or several projects if his budget allows. If the project succeeds, he receives a payoff from the project. When more investors invest in the same project, the succeeding probability of the project increases, which benefits all investors investing this project. Note that in both voting and investment games, the agent's decision has a positive effect on ones' decisions. When one agent makes a decision, the subsequent agents are encouraged to make the same decision in two aspects: the probability that this action has the positive outcome increases due to this agent's decision, and the potential reward of this action may be significantly large according to the belief of this agent.

The combination of negative network externality with social learning, on the other hand, is difficult to analyze. When the network externality is negative, the game becomes an anti-coordination game, where one agent seeks the strategy that differs from others' to maximize his own reward. Nevertheless, in such a scenario, the agent's decision also contains some information about his belief on the uncertain system state, which can be learned by subsequent agents through social learning algorithms. Thus, subsequent agents may then realize that his choice is better than others, and make the same decision with the agent. Since the network externality is negative, the information leaked by the agent's decision may impair the reward the agent can obtain in the game. Therefore, rational agents should take into account the possible reactions of subsequent players to maximize their own rewards.

The negative network externality plays an important rule in many applications in different research fields, such as spectrum access in cognitive radio, storage service selection in cloud computing, and deal selection on Groupon in online social networking. In spectrum access problem, for instance, secondary users access the same spectrum need to share with each others. The more secondary users access the same channel, the less available access time for each of them. In storage service selection problem, the reliability and availability are affected by the number of subscribers. The more subscribers using the same service, the lower the service quality of the cloud storage platform. For the deal selection on Groupon website, some businesses may receive overwhelming number of customers under the discounted deal. The overwhelming number of customers has a negative network externality on the quality of the products. In these examples, the negative network externality degrades the utility of the agents making the same decision. Therefore, the agents should take into account the possibility of degraded utility, e.g., less access time, lower reliability, or lower service quality, when making the decisions.

%In summary, when we are dealing with resource sharing problem with sequential decision structure, the interaction between negative network externality and learning effect becomes the main issue we need to address.

The aforementioned social learning approaches are mostly strategic, where agents are considered as players with bounded or unbounded rationality in maximizing their own rewards. Machine learning, which is another class of approaches for the learning problem, focuses on designing algorithms for making use of the past experience to improve the performance of similar tasks in the future \cite{mitchell1997machine}. Generally there exists some training data and the devices follow a learning method designed by the system designer to learn and improve the performance of some specific tasks. Most learning approaches studied in machine learning are non-strategic without the rationality on considering their own benefit. Such non-strategic learning approaches may not be applicable to the scenario where devices are rational and intelligent enough to choose actions to maximize their own benefits instead of following the rule designed by the system designer.

Chinese restaurant process, which is introduced in non-parametric learning methods in machine learning \cite{aldous1985exchangeable}, provides an interesting non-strategic learning method for unbounded number of objects. In Chinese restaurant process, there exists infinite number of tables, where each table has infinite number of seats. There are infinite number of customers entering the restaurant sequentially. When one customer enters the restaurant, he can choose either to share the table with other customers or to open a new table, with the probability being predefined by the process. Generally, if a table is occupied by more customers, then a new customer is more likely to join the table, and the probability that a customer opens a new table can be controlled by a parameter \cite{pitman1995exchangeable}. This process provides a systematic method to construct the parameters for modeling unknown distributions.

By introducing the strategic behavior into the non-strategic Chinese restaurant process, we proposed a new game, called \textbf{Chinese Restaurant Game}, to formulate the social learning problem with negative network externality. Let us consider a Chinese restaurant with $K$ tables. There are $N$ customers sequentially requesting for seats from these $K$ tables for having their meals. One customer may request one of the tables in number. After requesting, he will be seating in the table he requested. We assume that all customers are rational, i.e., they prefer bigger space for a comfortable dining experience. Thus, one may be delighted if he has a bigger table. However, since all tables are available to all customers, he may need to share the table with others if multiple customers request for the same table. In such a case, the customer's dining space reduces, due to which the dining experience is impaired. Therefore, the key issue in the proposed Chinese restaurant game is how the customers choose the tables to enhance their own dining experience. This model involves the negative network externality since the customer's dining experience is impaired when others share the same table with him. Moreover, when the table size is unknown to the customers, but each of them receives some signals related to the table size, this game involves the learning process if customers can observe previous actions or signals. Such a theoretic Chinese restaurant game framework is very general and can be applied into many research areas, such as online social networks, wireless communication, and cloud computing, which will be discussed in Part II of this two-part paper \cite{wang2011crgpart2}.

In the rest of the paper, we first provide detailed descriptions on the system model of Chinese restaurant game in Section \ref{sec_sys}. Then, we analyze the simultaneous game model to show how customers behave given the perfect knowledge on the table size in Section \ref{sec_simul}. Next, we study the sequential game model with perfect information to illustrate the advantage of playing first in Section \ref{sec_seque}. In Section \ref{sec_gen}, we show the general Chinese restaurant game framework by analyzing the learning behaviors of customers under the negative network externality and uncertain system state. We provide a recursive method to construct the best response for customers, and discuss the simulation results in Section \ref{sec_sim}. Finally, we draw conclusions in Section \ref{sec_con}.

\section{Related Works}
A closely-related strategic game model to our work is the global game \cite{carlsson1993global,morris2001global}. In the global game, all agents, with limited knowledge on the system state and information hold by other agents, make decisions simultaneously. The agent's reward in the game is determined by the system state and the number of agents making the same decision with him. The influence may be positive or negative depending on the type of network externality. An important characteristics of global game is that the equilibrium is unique, which simplifies the discussion on the possible outcome of the game. It draws great attentions in various research fields, such as financial crisis \cite{angeletos2006crises}, sensor networks \cite{krishnamurthy2011decentra} and cognitive radio networks \cite{krishnamurthy2009decentra}. Since all players in the global game make decisions simultaneously, there is no learning involved in the global game.

In recent years, several works \cite{dasgupta2000social,angeletos2006signaling,angeletos2007dynamic,costain2007herding,dasgupta2007coordination} make efforts to introduce the learning and signaling into the global game. Dasgupta's first attempt was investigating a binary investment model, while one project will succeed only when enough number of agents invest in the project in \cite{dasgupta2000social}. Then, Dasgupta studied a two-period dynamic global game, where the agents have the options to delay their decisions in order to have better private information of the unknown state in \cite{dasgupta2007coordination}.

Angeletos \emph{et. al.} studied a specific dynamic global game called regime of changes \cite{angeletos2006signaling,angeletos2007dynamic}. In the regime of changes game, each agent may propose an attack to the status quo, i.e., the current politic state of the society. When the collected attacks are large enough, the status quo is abandoned and all attackers receive positive payoffs. If the status quo does not change, the attackers receive negative payoffs. Angeletos \emph{et. al.} first studied a signaling model with signals at the beginning of the game in \cite{angeletos2006signaling}. Then, they proposed a multiple stages dynamic game to study the learning behaviors of agents in the regime of change game in \cite{angeletos2007dynamic}.

Costain provided a more general dynamic global game with an unknown binary state and a general utility function in \cite{costain2007herding}. The utility function includes information revelation, strategic complementarities, and payoff heterogeneity. To simplify the analysis, the positions of the agents in the game are assumed to be unknown. Nevertheless, most of these works study the multiplicity of equilibria in dynamic global game with simplified models, such as binary state or binary investment model. Moreover, the network externality they considered in their models are mostly positive. By proposing the Chinese restaurant game, we hereby provides a more general game-theoretic framework on studying the social learning in a network with negative network externality, which has many applications in various research fields.

%In machine learning, a non-strategic but related problem is the multi-armed bandit problem \cite{berry1985bandit}. In the multi-armed bandit problem, a bandit with multiple arms is provided to a gambler. The gambler may have different levels of rewards by playing different arms each time. Thus, the gambler may try and learn in each play in order to maximize his collected rewards. Liu and Zhao extended this model by considering multiple agents and including the network externality in \cite{liu2010distributed}. They studied how agents learn the expected payoff and other agent's choice by estimating the regrets after choosing different arms based on his current belief. Nevertheless, the study in multi-armed bandit problem is generally non-strategic. They assume agents will always follow the learning rule designed by the system designer. The strategic study on the learning problem in the proposed Chinese restaurant game is necessary when the agents are considered fully-rational.

\section{System Model}\label{sec_sys}
Let us consider a Chinese restaurant with $K$ tables numbered $1,2,...,K$ and $N$ customers labeled with $1,2,...,N$. Each customer requests for one table for having a meal. There may be multiple customers request for the same table. Each table has infinite seats, but may be in different size. We model the table sizes of a restaurant with two components: the restaurant state $\theta$ and the table size functions $\{R_1(\theta),R_2(\theta),...,R_K(\theta)\}$. The state $\theta$ represents an objective parameter, which may be changed when the restaurant is remodeled. The table size function $R_j(\theta)$ is fixed, i.e., the functions $\{R_1(\theta),R_2(\theta),...,R_K(\theta)\}$ will be the same every time the restaurant is remodeled. An example of $\theta$ is the order of existing tables. Suppose that the restaurant has two tables, one is of size $L$ and the other is of size $S$. Then, the owner may choose to number the large one as table $1$, and the small one as table $2$. The decision on the numbering can be modeled as $\theta \in \{1,2\}$, while the table size functions $R_1(\theta)$ and $R_2(\theta)$ are given as $R_1(1)=L$, $R_1(2)=S$, and $R_2(1)=S$, $R_2(2)=L$. Let $\Theta$ be the set of all possible state of the restaurant. In this example, $\Theta=\{1,2\}$.

We formulate the table selection problem as a game, called \textbf{Chinese Restaurant Game}. We first denote $\mathbf{X} = \{1,...,K\}$ as the action set (tables) that a customer may choose, where $x_i \in A$ means that customer $i$ chooses the table $x_i$ for a seat. Then, the utility function of customer $i$ is given by $U(R_{x_i},n_{x_i})$, where $n_{x_i}$ is the number of customers choosing table $x_i$. According to our previous discussion, the utility function should be an increasing function of $R_{x_i}$, and a decreasing function of $n_{x_i}$. Note that the decreasing characteristic of $U(R_{x_i},n_{x_i})$ over $n_{x_i}$ can be regarded as the negative network externality effect since the degradation of the utility is due to the joining of other customers. Finally, let $\mathbf{n}=\{n_1,n_2,...,n_K\}$ be the numbers of customers on the $K$ tables, i.e., the grouping of customers in the restaurant.

As mentioned above, the restaurant is in a state $\theta \in \Theta$. However, customers may not know the exact state $\theta$, i.e., they may not know the exact size of each table before requesting. Instead, they may have received some advertisements or gathered some reviews about the restaurant. The information can be treated as some kinds of signals related to the true state of the restaurant. In such a case, they can estimate $\theta$ through the available information, i.e., the information they know and/or gather in the game process. Therefore, we assume that all customers know the prior distribution of the state information $\theta$, which is denoted as $\mathbf{g_0} =\{g_{0,l}|g_{0,l}=Pr(\theta=l),~\forall l \in \Theta\}$. The signal each customer received $s_i \in S$ is generated from a predefined distribution $f(s|\theta)$.

\subsection{Belief on State}
 In this subsection, we introduce the concept of belief to describe how the customers estimate the system state $\theta$. Since customers may make decisions sequentially, it is possible that the customers who make decisions later learn the signals from those customers who make decisions first. Let us denote the signals customer $i$ learned, excluding his own signal $s_i$, as $\mathbf{h_i}=\{s\}$. With the help of these signals $\mathbf{h_i}$, his own signal $s_i$, the prior distribution $\mathbf{g_0}$, and the conditional distribution $f(s|\theta)$, each customer $i$ can estimate the current system state in probability with the belief being defined as
\begin{equation}
    \mathbf{g_i} = \{g_{i,l}|g_{i,l}=Pr(\theta=l|\mathbf{h_i},s_i,\mathbf{g_0},f),~\forall l \in \Theta\} ~\forall i \in N .
\end{equation}

According to the above definition, $g_{i,l}$ represents the probability that system state $\theta$ is equal to $l$ conditioning on the collected signals $\mathbf{h_i}$, received signal $s_i$, the prior probability $\mathbf{g_0}$, and the conditional distribution $f(s|\theta)$.
Notice that in the social learning literature, the belief can be obtained through either non-Bayesian naive updating rule \cite{bala1998learning,golub2007naive} or fully rational Bayesian rule \cite{acemoglu2011bayesian}. For the non-Bayesian naive updating rule, it is implicitly based on the assumption that customers are only limited rational and follows some predefined rules to compute their believes. Their capability to maximize their utilities is limited not only by the game structure and learned information, but also by the non-Bayesian naive updating rules. In the fully rational Bayesian rule, customers are fully rational and have the potential to optimize their actions without the restriction on the fixed belief updating rule.
Since the customers we considered here are fully rational, they will follow the Bayesian rule to update their believes as follows:
\begin{equation}\label{eqn_g}
    g_{i,l}= \frac{g_{0,l}Pr(\mathbf{h_i},s_i|\theta=l)}{\sum_{l' \in \Theta}g_{0,l'}Pr(\mathbf{h_i},s_i|\theta=l')}.
\end{equation}
Notice that the exact expression for belief updating depends on how the signals are generated and learned, which is generally affected by the conditional distribution $f(s|\theta)$ and the game structure.

\section{Simultaneous Game with Perfect Signal: How Negative Network Externality Affects}\label{sec_simul}
The first game structure we would like to discuss is the simultaneous game, in which all customers make decisions simultaneously, e.g., all agents arrive the restaurant at the same time. In such a scenario, there is no learning involved in the game since customers request the tables at the same time. By investigating this game model, we can have an initial understanding on how customers behave in the game.

We start with a simple case where there are only two customers and two tables. In such a case, there are two possible system states, $\Theta=\{\theta_1, \theta_2\}$, indicating which table is larger. When the system state is $\theta_1$, $R_1(\theta_1)=L$ and $R_2(\theta_1)=S$ where $L \geq S$. On the other hand, if the system state is $\theta_2$, then $R_1(\theta_2)=L$ and $R_2(\theta_2)=S$. Moreover, in such a scenario, the signal is assumed to perfectly reveal the system state and indicate the exact amount of resource in each pool, e.g., $S=\{s_1,s_2\},~f(s_1|\theta=\theta_1)=1$, and $~f(s_2|\theta=\theta_2)=1$. Under such a signal structure, a customer can immediately know what the system state is when receiving the signal. The customer also knows that his opponent has the perfect signal of the true system state, i.e., the opponent also knows the exact size of each table. With such a simple setting, we temporally remove the uncertainty on the state to see how customers make decisions given the network externality.

Given the opponent's decision, which may be the larger table or smaller table, a rational customer should choose the table that can maximize his utility. If the decision made by the opponent is the smaller one, then the customer's best choice is the larger one since $U(L,1) > U(S,1)$. However, if the opponent chooses the larger table, the customer's choice depends on how severe the negative network externality is. If $U(L,2) > U(S,1)$, the customer will choose the same larger table. Otherwise, the smaller table will be chosen. The result shows that even both customers have already learned the true system state, i.e., which table is larger, they may not both choose the larger one in the equilibrium. Instead, the equilibrium depends on how severe the negative network externality is. If the externality results in an unacceptable penalty, then customers should choose different tables to avoid it.

\subsection{Best Response of Customers Under Perfect Signal}
Now let us consider the general scenario where there are $N$ customers, $K$ tables, and $L$ possible states $\Theta=\{\theta_1,...,\theta_L\}$. Here, we consider the perfect signals case, i.e., the system state $\theta$ and the sizes of tables $R_1(\theta),R_2(\theta),...,R_K(\theta)$ are known by all customers. The imperfect signal case will be discussed in Section \ref{sec_gen}. Since the customers are rational, their objectives in this game are to maximize their own utilities. However, since their utilities are determined by not only their own actions but also others', the customers' behaviors in the game are influenced by each other.

A strategy describes how a player will play given any possible situations in the game. In the simultaneous Chinese restaurant game, the customer's strategy should be a mapping from other customers' table selections to his own table selection. Recalling that $n_j$ stands for the number of customers choosing table $j$. Let us denote $\mathbf{n_{-i}}=\{n_{-i,1},n_{-i,2},...,n_{-i,K}\}$ with $n_{-i,j}$ being the number of customers except customer $i$ choosing table $j$. Then, given $\mathbf{n_{-i}}$, a rational customer $i$ should choose the action as
\begin{equation}\label{eqn3}
    BE_i(\mathbf{n_{-i}},\theta) = \arg \max_{x \in A} U(R_x(\theta),n_{-i,x}+1).
\end{equation}

The (\ref{eqn3}) describes a special set of strategy called best response, which represents the optimal action of a customer that maximizes the utility given other customers' actions. In the following, we give a formal definition of best response.
\begin{dfn} Considering a game with $N$ players, each with an action space $A_i$ and a utility function $u_i(x_i,x_{-i})$, where $x_i$ is player $i$'s action and $x_{-i}$ is the actions of all players except player $i$. The best response of customer $i$ is
\begin{equation}
    BE_i(x_{-i}) = \arg \max_{x \in A_i} u_i(x,x_{-i}).
\end{equation}
\end{dfn}

\subsection{Nash Equilibrium Under Perfect Signal}
Nash equilibrium is a popular concept for predicting the outcome of a game with rational customers. Informally speaking, Nash equilibrium is an action profile, where each customer's action is the best response to other customers' actions in the profile. Since all customers use their best responses, none of them have the incentive to deviate from their actions. A formal definition of Nash equilibrium for the simultaneous game is given as follows.

\begin{dfn}[Nash Equilibrum] %Considering a game with players $1,2,...,N$. Each player $i$ has an action space $A_i$ and a utility function $u_i(a_i,a_{-i})$, where $a_i$ is the player's action and $a_{-i}$ is the action profile of all players except player $i$.
Nash equilibrium is the action profile $\mathbf{x^*}=\{x_1^*,x_2^*,...,x_N^*\}$ where $\forall i \in N$, $BE_i(x^*_{-i})=x_i^*$.
\end{dfn}

According to the definition of Nash equilibrium, the sufficient and necessary condition of Nash equilibrium in the simultaneous Chinese restaurant game is stated in the following theorem.

\begin{thm}\label{thm_ne_sim_p} Given the customer set $\{1,...,N\}$, the table set $\{1,...,K\}$, and the current system state $\theta$, for any Nash equilibrium of the simultaneous Chinese restaurant game with perfect signal, its equilibrium grouping $\mathbf{n^*}$ should satisfy the following conditions
\begin{equation}\label{eq_ne_sim_p}
    U(R_x(\theta),n^*_x) \geq U(R_y(\theta),n^*_y+1),\ \ if\ n^*_x > 0, \forall x,y \in \{1,...,K\}.
\end{equation}
\end{thm}
\begin{proof}

\begin{itemize}
  \item Sufficient condition: suppose that the action profile of all players is $\mathbf{x}=\{x_1,...,x_N\}$ and such an action profile leads to the grouping $\mathbf{n^*}=\{n_1^*,...,n_K^*\}$ that satisfies (\ref{eq_ne_sim_p}). Without loss of generality, let us assume that customer $i$ chooses table $j$, i.e., $x_i=j$, then we have
      \begin{equation}
           u_i(x_i,x_{-i}) = U(R_j(\theta),n^*_j).
      \end{equation}
      If customer $i$ chooses any other table $k\neq j$, i.e., $x'_i=k\neq x_i=j$, then his utility becomes
      \begin{equation}
           u'_i(x'_i,x_{-i})  = U(R_k(\theta),n^*_k+1).
      \end{equation}
      Since $U(R_j(\theta),n^*_j)\geq U(R_k(\theta),n^*_k+1), \forall j, k$, we have $BE_i(x_{-i})=x_i, \forall i$. Therefore, $a=\{x_1,...,x_N\}$ is a Nash equilibrium.
  \item Necessary condition: suppose that the Nash equilibrium $\mathbf{x^*}=\{x_1^*,x_2^*,...,x_N^*\}$ leads to the grouping $\mathbf{n^*}=\{n_1^*,...,n_K^*\}$. Without loss of generality, let us assume that customer $i$ chooses table $j$, i.e., $x_i^*=j$, then we have
      \begin{equation}
           u_i(x_i^*,x^*_{-i}) = U(R_j(\theta),n^*_j) \text{ and } n^*_j > 0.
      \end{equation}
      If customer $i$ chooses any other table $k\neq j$, i.e., $x'_i=k\neq x_i^*=j$, then his utility becomes
      \begin{equation}
           u'_i(x'_i,x^*_{-i})  = U(R_k(\theta),n^*_k+1).
      \end{equation}
      Since $\mathbf{x^*}=\{x_1^*,x_2^*,...,x_N^*\}$ is a Nash equilibrium, we have $BE_i(x^*_{-i})=x_i^*, \forall i$, i.e.,
      \begin{equation}
    U(R_j(\theta),n^*_j) \geq U(R_k(\theta),n^*_k+1),\ \ if\ n^*_j > 0, \forall j,k \in \{1,...,K\}.
\end{equation}
\end{itemize}
\end{proof}

From Theorem \ref{thm_ne_sim_p}, we can see that, at Nash equilibrium, one customer's utility would never become higher by deviating to another table. Moreover, any deviation to another table will degrade the utility of all customers in that table due to the negative network externality. The (\ref{eq_ne_sim_p}) also implies that customers may eventually have different utilities even the tables they choose have the same size. A simple example would be a three-customer restaurant with two tables in exact same size. Since there are three customers, at the Nash equilibrium, one of the table must be chosen by two customers while the other table is occupied only by one customer.

\subsection{Uniqueness of Equilibrium Grouping}
Obviously, there will be more than one Nash equilibrium since we can always exchange the actions of any two customers in one Nash equilibrium to build a new Nash equilibrium without violating the sufficient and necessary condition shown in (\ref{eq_ne_sim_p}). Nevertheless, the equilibrium grouping $\mathbf{n^*}$ may be unique as stated in the following Theorem.

\begin{thm}\label{thm_exist_unique_ne_sim_p} There exists a Nash equilibrium in the simultaneous game with perfect signal. If the inequality in (\ref{eq_ne_sim_p}) strictly holds for all $x,y \in \{1,...,K\}$, then the equilibrium grouping $\mathbf{n}^*=\{n_1^*,...,n_K^*\}$ is unique.
\end{thm}
\begin{proof} Since the signals are prefect, the current system state $\theta$ is known by all customers. In the following, we first propose a greedy algorithm to construct a Nash equilibrium and then show the uniqueness of the equilibrium grouping when the inequality strictly holds. Since exchanging the actions of any two customers at the Nash equilibrium will lead to another Nash equilibrium, the Nash equilibrium is generally not unique. Without loss of generality, in the proposed greedy algorithm, the customers choose their actions sequentially with customer $i$ being the $i-th$ customer choosing the action. We let customers choose their actions in the myopic way, i.e., they choose the tables that can maximize their current utilities purely based on what they have observed. Let $\mathbf{n_i}=\{n_{i,1},n_{i,2},...,n_{i,K}\}$ with $\sum_{j=1}^{K} n_{i,j}=i-1$ be the grouping observed by customer $i$. Then, customer $i$ will choose the myopic action given by
\begin{equation}
    BE^{myopic}_i(\mathbf{n_i},\theta)=\arg \max_{x \in A} U(R_x(\theta),n_{i,x}+1).
\end{equation}

Let $\mathbf{x^*}=\{x_1^*,x_2^*,...,x_N^*\}$ be the output action set of the proposed greedy algorithm and $\mathbf{n^*}=\{n_1^*,n_2^*,...,n_K^*\}$ be the corresponding grouping. For any table $j$ with $n^*_j > 0$, suppose customer $k$ is the last customer choosing table $j$. Then, we have
\begin{equation}\label{eqn12}
    U(R_j(\theta),n_{k,j}+1) \geq U(R_{j'}(\theta),n_{k,j'}+1), \forall j' \in \{1,...,K\}.
\end{equation}

Since customer $k$ is the last customer choosing table $j$, we have $n^*_j=n_{k,j}+1$ and $n^*_{j'} \geq n_{k,j'}$. Then, according to (\ref{eqn12}), $\forall j' \in \{1,...,K\}$, we have
\begin{equation}\label{eqn13}
    U(R_j(\theta),n^*_{j})= U(R_j(\theta),n_{k,j}+1) \geq U(R_{j'}(\theta),n_{k,j'}+1) \geq U(R_{j'}(\theta),n^*_{j'}+1),
\end{equation}
where the last inequality comes from the fact that $U(\cdot)$ is a decreasing function in terms of $n$.

Note that (\ref{eqn13}) holds for all $j,j' \in \{1,...,K\}$ with $n^*_j > 0$, i.e., $U(R_j(\theta),n^*_{j}) \geq U(R_{j'}(\theta),n^*_{j'}+1),\ \forall j,j' \in \{1,...,K\}\ with\ n^*_j > 0$. According to Theorem \ref{thm_ne_sim_p} and (\ref{eqn13}), the output action set $\mathbf{x^*}=\{x_1^*,x_2^*,...,x_N^*\}$ from the proposed greedy algorithm is a Nash equilibrium.

Next, we would like to prove by contradiction that if the inequality in (\ref{eq_ne_sim_p}) strictly holds, the equilibrium grouping $\mathbf{n}^*=\{n_1^*,...,n_K^*\}$ is unique. Suppose that there exists another Nash equilibrium with equilibrium goruping $\mathbf{n'}=\{n'_1,...,n'_K\}$, where $n'_j \neq n^*_j$ for some $j \in \{1,...,K\}$. Since both $\mathbf{n}^*$ and $\mathbf{n'}$ are equilibrium groupings, we have $\sum_{j=1}^K n'_j = \sum_{j=1}^K n^*_j = N$. In such a case, there exists two table $x$ and $y$ with $n'_x > n^*_x$ and $n'_y < n^*_y$. Then, since $\mathbf{n^*}$ is an equilibrium grouping, we have
\begin{equation}\label{eqn15}
    U(R_y(\theta),n^*_y) > U(R_x(\theta),n^*_x+1).
\end{equation}

Since $n'_x > n^*_x$, $n'_y < n^*_y$, and $U(\cdot)$ is a deceasing function of $n$, we have
\begin{equation}\label{eqn17}
    U(R_x(\theta),n^*_x) > U(R_x(\theta),n^*_x+1) \geq U(R_x(\theta),n'_x),
\end{equation}
and
\begin{equation}\label{eqn18}
    U(R_y(\theta),n'_y) > U(R_y(\theta),n'_y+1) \geq U(R_y(\theta),n^*_y).
\end{equation}

Since $\mathbf{n'}$ is also an equilibrium grouping, we have
\begin{equation}\label{eqn19}
    U(R_x(\theta),n'_x) \geq U(R_y(\theta),n'_y+1).
\end{equation}

According to (\ref{eqn17}), (\ref{eqn18}), and (\ref{eqn19}) we have
\begin{equation}
    U(R_x(\theta),n^*_x+1) \geq U(R_x(\theta),n'_x) \geq U(R_y(\theta),n'_y+1) \geq U(R_y(\theta),n^*_y),
\end{equation}
which contradicts with (\ref{eqn15}). Therefore, the equilibrium grouping $\mathbf{n}^*$ is unique when the inequality in (\ref{eq_ne_sim_p}) strictly holds.
\end{proof}

\section{Sequential Game with Perfect Signal: The Advantage of Playing First}\label{sec_seque}
In the previous section, we have studied how negative network externality affects the action of each customer and found that a balance will finally be achieved among tables such that there will be no overwhelming requests for one table. However, we also find that some customers may have higher utilities at the Nash equilibrium. In this section, we extend the Chinese restaurant game into a sequential game model, where customers choose their actions in a pre-determined order. We assume that every customer can observe all the actions chosen before him, but cannot change the action once chosen.

%Let us consider the Chinese restaurant game with $K$ tables and $N$ customers. Since the signals are prefect, all customers know exactly what the current system state $\theta$ is.
In this sequential Chinese restaurant game, customers make decisions sequentially with a predetermined order known by all customers, e.g., waiting in a line of the queue outside of the restaurant. Without loss of generality, in the rest of this paper, we assume the order is the same as the customer's number, i.e., the order of customer $i$ is $i$. We assume every customer knows the decisions of the customers who make decisions before him, i.e., customer $i$ knows the decisions of customers $\{1,...,i-1\}$.  Let $\mathbf{n_i}=\{n_{i,1},n_{i,2},...,n_{i,K}\}$ be the current grouping, i.e., the number of customers choosing table $\{1,2,...,K\}$ before customer $i$. The $\mathbf{n_i}$ roughly represents how crowded each tables is when customer $i$ enters the restaurant. Notice that $\mathbf{n_i}$ may not be equal to $\mathbf{n}$, which is the final grouping that determines customers' utilities. A table with only few customers may eventually be chosen by many customers at the end of the game. %We may consider $\mathbf{n_i}$ as a reference of current popularity of tables.
The best response function of customer $i$ is
\begin{equation}
    BE^{se}_i(\mathbf{n_i},\theta)=\arg \max_x U(R_x(\theta),n_x(\mathbf{n_i})),
\end{equation}
where $n_x(\mathbf{n_i})$ denotes the expected number of customers choosing table $x$ given $\mathbf{n_i}$. The problem here is how to predict the decisions of the remaining customers given the current observation $\mathbf{n_i}$ and state $\theta$.

\subsection{Subgame Perfect Nash Equilibrium and Advantage of Playing First}
In this subsection, we will study the possible equilibria of the sequential Chinese restaurant game. In particular, we will study the subgame perfect Nash equilibrium. A subgame is a part of the original game. In our sequential Chinese restaurant game, any game process begins from player $i$, given all possible actions before player $i$, could be a subgame. A formal definition of subgame is given as follows.

\begin{dfn} A subgame in the sequential Chinese restaurant game is consisted of two elements: 1) It begins from customer $i$; 2) The current grouping before customer $i$ is $\mathbf{n_i}=\{n_{i,1},...,n_{i,K}\}$ with $\sum_{j=1}^{K} n_{i,j} = i-1$.
\end{dfn}

With the definition of subgame, a subgame perfect Nash equilibrium is defined as follows.

\begin{dfn} A Nash equilibrium is a subgame perfect Nash equilibrium if and only if it is a Nash equilibrium for any subgame.
\end{dfn}

%Informally speaking, a Nash equilibrium can be a subgame perfect Nash equilibrium if and only if, given any possible subgame, the best response strategy in the original game is still the best response strategy in the subgame.
With the concept of subgame perfect Nash equilibrium, we can refine the number of Nash equilibria in the original game. We would like to show the subgame perfect Nash equilibrium in the sequential Chinese restaurant game by constructing one. Given a subgame, the corresponding equilibrium grouping and the best responses of agents in the subgame can be derived through two functions as follows. Let $EG(\mathbf{X_s},N_s)$ be the function that generates the equilibrium grouping for a table set $\mathbf{X_s}$ and number of customers $N_s$. The equilibrium grouping is given by (\ref{eq_ne_sim_p}) with $\mathbf{X}$ being replaced by $\mathbf{X_s}$ and $N$ being replaced by $N_s$. Notice that $\mathbf{X_s}$ could be any subset of the total table set $\{1,...,K\}$, and $N_s$ is less or equal to $N$. We will prove that the output of $EG(\cdot)$ is the corresponding equilibrium grouping in the subgame in Lemma \ref{lem_eg_subgame}. Then, let $PC(\mathbf{X_s},\mathbf{n_s},N_s)$, where $\mathbf{n_s}$ denotes the current grouping observed by the customers, be the algorithm that generates the set of available tables given the current grouping $\mathbf{n_s}$ in the subgame. The algorithm removes the tables which have been already over-requested, i.e., the tables that already occupied by more than the expected number of customers in the equilibrium grouping. The procedures of implementing $PC(\mathbf{X_s},\mathbf{n_s},N_s)$ are described as follows:
\begin{enumerate}
    \item Initialize: $\mathbf{X_o} = \mathbf{X_s}$, $N_t = N_s$
    \item $\mathbf{X_t} = \mathbf{X_o}$, $\mathbf{n^e}=EG(\mathbf{X_t},N_t)$, $\mathbf{X_o} = \{x|x \in \mathbf{X_t}, n^e_j \geq n_{s,j}\}$, $N_t = N_s - \sum_{x \in \mathbf{X_s} \setminus \mathbf{X_o}} n_{s,x}$.
    \item If $\mathbf{X_o} \neq \mathbf{X_t}$, go back to step 2.
    \item Output $\mathbf{X_o}$.
\end{enumerate}
As shown in the following Lemma, the $PC(\mathbf{X_s},\mathbf{n_s},N_s)$ will never remove the tables that are best choices of the customers.

\begin{lem}\label{lem_pg_remove} Given a subgame with current grouping $\mathbf{n_s}$, current available table set $\mathbf{X_s}$, and the number of players $N_s$, if table $j \not\in \mathbf{X'_s}=PC(\mathbf{X_s},\mathbf{n_s},N_s)$, then there exists at least one table $j' \in \mathbf{X'_s}$ such that $U(R_{j'}(\theta),n'_{j'}) \geq U(R_{j}(\theta),n_{s,j})$.%, where $\mathbf{n^*}$ is the grouping in the Nash equilibrium of this subgame.
\end{lem}
\begin{proof}%According to the decreasing characteristic of $U(R,n)$ with respect to $n$, we have
%\begin{equation}
%    U(R_j(\theta),n^*_j) > U(R_j(\theta),n_{s,j}).
%\end{equation}

Let $\mathbf{n^*} = EG(N_s,X_s)$. Since table $j$ is removed by $PC(\mathbf{X_s},\mathbf{n_s},N_s)$, the inequality $n_{s,j} > n^*_j$ should be hold. However, since $n_{s,j} > n^*_j$, the equilibrium grouping $\mathbf{n^*}$ is impossible to be reached in the Nash equilibrium of this subgame. Assuming the Nash equilibrium of this subgame is $\mathbf{n'}$, we have $n'_j \geq n_{s,j} > n^*_j$, which means $\sum_{k \in X_s, k \neq j} n'_k < \sum_{k \in X_s, k \neq j} n^*_k$. Therefore, there exists a $j' \in \mathbf{X_s} \setminus \{j\}$ such that $n'_{j'} < n^*_{j'}$.

Since $\mathbf{n^*}$ is an equilibrium grouping, we have
\begin{equation}\label{eqn23}
    U(R_{j'}(\theta),n^*_{j'}) \geq U(R_{j}(\theta),n^*_{j}+1).
\end{equation}

According to the above discussions, we have
\begin{equation}\label{eqn24}
U(R_{j'}(\theta),n'_{j'}) > U(R_{j'}(\theta),n^*_{j'}) \geq  U(R_{j}(\theta),n^*_{j}+1) \geq U(R_j(\theta),n_{s,j}) \geq U(R_j(\theta),n'_j)
    %U(R_j(\theta),n'_j) \leq U(R_j(\theta),n_{s,j}) \leq U(R_{j}(\theta),n^*_{j}+1) \leq U(R_{j'}(\theta),n^*_{j'}) < U(R_{j'}(\theta),n'_{j'}),
\end{equation}
where the the first inequality is due to $n'_{j'} < n^*_{j'}$, the second inequality is due to (\ref{eqn23}), and the last two inequalities are due to $n'_j \geq n_{s,j} > n^*_j$.

 According to (\ref{eqn24}), there exists at least one table $j'$ that can give the customer a higher utility than table $j$ in this subgame. Therefore, table $j$ is never the best response of the customer.
\end{proof}

Now, we propose a method to construct a subgame perfect Nash equilibrium. This equilibrium will satisfy the equilibrium grouping in (\ref{eq_ne_sim_p}). For each customer $i$, his strategy is described as follows:
\begin{equation}\label{eq_be_seq}
    BE^{se*}_i(\mathbf{n_i},\theta)=\arg \max_{x \in \mathbf{X^{i,cand}},n_{i,x} < n^{i,cand}_x} U(R_x(\theta),n^{i,cand}_x),
\end{equation}
where
\begin{eqnarray}
    \mathbf{X^{i,cand}} = PC(\mathbf{X},\mathbf{n_i},N), \\
    N^{i,cand} = N - \sum_{x \in X \setminus X^{i,cand}} n_{i,x}, \\
    \mathbf{n^{i,cand}} = EG(\mathbf{X^{i,cand}},N^{i,cand}).
\end{eqnarray}

In Lemma \ref{lem_eg_subgame}, we show that the above strategy results in the equilibrium grouping in any subgame.
\begin{lem}\label{lem_eg_subgame} Given the available table set $X_s=PC(\mathbf{X},\mathbf{n_s},N)$, $N_s = N - \sum_{x \in X \setminus X_s} n_{s,x}$, the proposed strategy shown in (\ref{eq_be_seq}) leads to an equilibrium grouping $\mathbf{n^*_s} = EG(\mathbf{X_s},N_s)$.% satisfying:
%\begin{equation}
%   U(R_y(\theta),n^s_y) \geq U(R_z(\theta),n^s_z+1), \ if\ n^s_y \neq 0, \forall y,z \in \mathbf{X_s}.
%\end{equation}
\end{lem}
\begin{proof}
We prove this by contradiction. Let $\mathbf{n}=\{n_j| j \in X_s\}$ be the final grouping after all customers choose their tables according to (\ref{eq_be_seq}). Suppose that $\mathbf{n} \neq \mathbf{n^*_s}=EG(X_s,N_s)$, then there exists some tables $j$ that $n_j > n^*_{s,j}$. Let table $j$ be the first table that exceeds $n_{s,j}$ in this sequential subgame. Since $n_j > n^*_{s,j}$, there are at least $n^*_{s,j}+1$ customers choosing table $j$. Suppose the $n^*_{s,j}+1$-th customer choosing table $j$ is customer $i$. Let $\mathbf{n_i}=\{n_{i,1},n_{i,2},...,n_{i,K}\}$ be the current grouping observed by customer $i$ before he chooses the table. Since customer $i$ is the $n^*_{s,j}+1$-th customer choosing table $j$, we have $n_{i,j}=n^*_{s,j}.$ Since table $j$ is the first table exceeding $\mathbf{n^*_s}$ after customer $i$'s choice, we have
\begin{equation}
    n_{i,x} \leq n^*_{s,x}~\forall x \in X_s.
\end{equation}

According to the definition of $PC(\cdot)$, none of the tables will be removed from candidates. Thus, $X^{i,cand}=X_s$ and $N^{i,cand}=N_s$. We have
\begin{equation}
    \mathbf{n^{i,cand}} = EG(\mathbf{X^{i,cand}},N^{i,cand}) = EG(X_s,N_s) = \mathbf{n^*_s}.
\end{equation}

However, according to (\ref{eq_be_seq}), the customer $i$ should not choose table $j$ since $n_{i,j}=n^*_{s,j}=n^{i,cand}_j$. This contradicts with our assumption that customer $i$ is the $n^*_{s,j}+1$-th customer choosing table $j$. Thus, the strategy (\ref{eq_be_seq}) should lead to the equilibrium grouping $\mathbf{n^*_s}=EG(X_s,N_s)$.
\end{proof}

Note that Lemma \ref{lem_eg_subgame} also shows that the final grouping of the sequential game should be $\mathbf{n^*}=EG(X,N)$ if all customers follow the proposed strategy in (\ref{eq_be_seq}). With Lemma \ref{lem_eg_subgame}, we show the existence of subgame perfect Nash equilibrium with the following Theorem.

\begin{thm}\label{thm_exist_ne_seq_p} Given customer set $\{1,...,N\}$, table set $\mathbf{X}=\{1,...,K\}$, and the current state $\theta$, there always exists a subgame perfect Nash equilibrium with the corresponding equilibrium grouping $\mathbf{n^*}=\{n_1^*,...,n_K^*\}$ satisfying the conditions in (\ref{eq_ne_sim_p}).
%\begin{equation}\label{eq_ne_seq_p}
%   U(R_y(\theta),n^*_y) \geq U(R_z(\theta),n^*_z+1), \ if\ n^*_y \neq 0, \forall y,z \in \{1,...,K\}.
%\end{equation}
\end{thm}

\begin{proof}
We would like to show that the proposed strategy in (\ref{eq_be_seq}) forms a Nash equilibrium. Suppose customer $i$ chooses table $j$ in his round according to (\ref{eq_be_seq}). Then, customer $i$'s utility is $u_i=U(R_j(\theta),n^*_j)$ since based on Lemma \ref{lem_eg_subgame}, the equilibrium grouping $\mathbf{n^*}$ will be reached at the end of the game.

If customer $i$ is the last customer, i.e, $i=N$, and chooses another table $j' \neq j$ in his round, then his utility becomes $U(R_{j'}(\theta),n^*_{j'}+1)$. However, according to (\ref{eq_ne_sim_p}), we have
\begin{equation}
    u^*_j = U(R_{j}(\theta),n^*_j) \geq U(R_{j'}(\theta),n^*_{j'}+1) .
\end{equation}
Thus, choosing table $j$ is never worse than choosing table $j'$ for customer $N$.

If that customer $i$ is not the last customer, and he chooses table $j'$ instead of table $j$ in his round. Since all customers before customer $i$ follows (\ref{eq_be_seq}), we have $n_{i,j} \leq n^*_{j}~\forall j \in \mathbf{X}$. Otherwise, $\mathbf{n^*}$ cannot be reached, which contradicts with Lemma \ref{lem_eg_subgame}.

If $n_{i,j'} < n^*_{j'}$, we have $n_{i+1,j'} \leq n^*_{j'}$. In addition, we have $n_{i+1,j} = n_{i,j} \leq n^*_{j}~\forall j \in \mathbf{X} \setminus \{j'\}$, since other tables are not chosen by customer $i$. Thus, $\mathbf{X^{i+1,cand}}=PC(X,\mathbf{n_{i+1}},N=X)$ and $N^{i,cand}=N$. According to Lemma \ref{lem_eg_subgame}, the final grouping should be $\mathbf{n^*}=EG(X,N)$. Thus, the new utility of customer $i$ becomes $u'_i = U(R_{j'}(\theta),n^*_{j'})$. However, according to (\ref{eq_be_seq}), we have
\begin{equation}
    u_i=U(R_j(\theta),n^*_j) =\arg \max_{x \in \mathbf{X},n_{i,x} < n^*_x} U(R_x(\theta),n^*_x) \geq  U(R_{j'}(\theta),n^*_{j'}) = u'_i.
\end{equation}

Thus, choosing table $j'$ never gives customer $i$ a higher utility.

If $n_{i,j'} = n^*_{j'}$, and the final grouping is $\mathbf{n'}=\{n'_1,n'_2,...,n'_K\}$. Since customer $i$ chooses table $j'$ when $n_{i,j'} = n^*_{j'}$, we have $n'_{j'} \geq n_{i+1,j'} = n_{i,j'}+1 = n^*_{j'}+1$. Thus, we have
\begin{equation}
    u_i = U(R_{j}(\theta),n^*_j) \geq U(R_{j'}(\theta),n^*_{j'}+1) \geq U(R_{j'}(\theta),n'_{j'}) = u'_i,~\forall j' \in X,
\end{equation}
where the first inequality comes from the equilibrium grouping condition in (\ref{eq_ne_sim_p}), and the second inequality comes from the fact that $U(R,n)$ is decreasing over $n$ and $n'_{j'} \geq n^*_{j'}+1$. Thus, under both cases, choosing table $j'$ is never better than choosing table $j$. We conclude that $\{BE^{se*}_i(\cdot)\}$ in (\ref{eq_be_seq}) forms a Nash equilibrium, where the grouping being the equilibrium grouping $\mathbf{n^*}$.

Finally, we show that the proposed strategy forms a Nash equilibrium in every subgame. In Lemma \ref{lem_pg_remove}, we show that if the table $j$ is removed by $PC(X,\mathbf{n_s},N)$, it is never the best response of all remaining customers. Thus, we only need to consider the remaining table candidates $X_s = PC(X,\mathbf{n_s},N)$ in the subgame. Then, with Lemma \ref{lem_eg_subgame}, we show that for every possible subgame with corresponding $\mathbf{X_s}$, the equilibrium grouping $\mathbf{n^*_s}=EG(X_s,N_s)$ will be achieved at the end of the subgame. Moreover, the above proof shows that if the equilibrium grouping $\mathbf{n_s}$ will be achieved at the end of the subgame, $BE^{se*}_i(\cdot)$ is the best response function. Therefore, the proposed strategies indeed form a Nash equilibrium in every subgame, i.e., we have a subgame perfect Nash equilibrium.
\end{proof}

In the proof of the subgame perfect Nash equilibrium, we observe that sequential game structure brings advantages for those customers making decisions at the beginning of the game. According to (\ref{eq_be_seq}), customers who make decisions early can choose the table that provides the largest utility in the equilibrium. When the number of customers choosing that table reaches equilibrium number, the second best table will be chosen by subsequent customers until it is full again. For the last customer, he has no choice but to choose the worst one.

\section{Imperfect Signal Model: How Learning Evolves}\label{sec_gen}
In Section \ref{sec_seque}, we have showed that in the sequential game with perfect signal, customers choosing first have the advantages for getting better tables and thus higher utilities. However, such a conclusion may not be true when the signals are not perfect. When there are uncertainties on the table sizes, customers who arrive first may not choose the right tables, due to which their utilities may be lower. Instead, customers who arrive later may eventually have better chances to get the better tables since they can collect more information to make the right decisions. In other words, when signals are not perfect, learning will occur and may result in higher utilities for customers choosing later. Therefore, there is a trade-off between more choices when playing first and more accurate signals when playing later. In this section, we would like to study this trade-off by discussing the imperfect signal model.

In the imperfect signal model, we assume that the system state $\theta \in \Theta=\{1,2,...,L\}$ is unknown to all $N$ customers. The sizes of $K$ tables can be expressed as functions of $\theta$, which are denoted as $R_1(\theta),R_2(\theta),...,R_K(\theta)$. The prior probability of $\theta$, $\mathbf{g_0}=\{g_{0,1},g_{0,2},...,g_{0,K}\}$ with $g_{0,l}=Pr(\theta=l)$, is assumed to be known by all customers. Moreover, each customer receives a private signal $s_i \in S$, which follows a p.d.f $f(s|\theta)$. Here, we assume $f(s|\theta)$ is public information to all customers. When conditioning on the system state $\theta$, the signals received by the customers are uncorrelated. %In addition, $f(s|\theta=1) \neq f(s|\theta=0)$ for some $s \in S$.

In this sequential Chinese restaurant game with imperfect signal model, the customers make decisions sequentially with the decision orders being their numbers. After a customer $i$ made his decision, he cannot change his mind in any subsequent time and his decision and signal are revealed to all other customers. We assume customers are fully rational, which means that they follow the Bayesian learning method to learn the true state and choose their strategies to maximize their own utilities.

Since signals are revealed sequentially, the customers who make decisions later can collect more information for better estimations of the system state. When a new signal is revealed, all customers follow the Bayesian rule to update their believes based on their current believes. Derived from (\ref{eqn_g}), we have the following Bayesian belief updating function
\begin{equation}\label{eqn33}
    g_{i,l} = \frac{g_{i-1,l}f(s_i|\theta=l)}{\sum_{w \in \Theta} {g_{i-1,w}f(s_i|\theta=w)}}.
\end{equation}
Based on the updating rule in (\ref{eqn33}), customer $i$ can update his belief when a new signal is revealed.

\subsection{Best Response of Customers}
Since the customers are rational, they will choose the action to maximize their own expected utility conditioning on the information they collect. Let $\mathbf{n_i} = \{n_{i,1},n_{i,2},...,n_{i,K}\}$ be the current grouping observed by customer $i$ before he chooses the table, where $n_{i,j}$ is the number of customers choosing table $j$ before customer $i$. Then, let $\mathbf{h_i}=\{s_1,s_2,...,s_{i-1}\}$ be the history of revealed signals before customer $i$. In such a case, the best response of customer $i$ can be written as
\begin{equation}\label{eqn34}
    x_i = BE_i(\mathbf{n_i},\mathbf{h_i},s_i) = \arg\max_j E[u_i(R_j(\theta),n_j)|\mathbf{n_i}, \mathbf{h_i}, s_i].
\end{equation}

From (\ref{eqn34}), we can see that when estimating the expected utility in the best response function, there are two key terms needed to be estimated by the customer: the system state $\theta$ and the final grouping $\mathbf{n}=\{n_1,n_2,...,n_K\}$. The system state $\theta$ is estimated using the concept of belief denoted as $\mathbf{g_i}=\{g_{i,1},g_{i,2},...,g_{i,L}\}$ with $g_{i,l}=Pr(\theta=l|\mathbf{h_i},s_i)$. Since the information on the system state $\theta$ in $\mathbf{n_i}$ is fully revealed by $\mathbf{h_i}$, given $\mathbf{h_i}$, $\mathbf{g_i}$ is independent with $\mathbf{n_i}$. Therefore, given the customer's belief $\mathbf{g_i}$, the expected utility of customer $i$ choosing table $j$ becomes
\begin{equation}\label{eqn35}
    E[u_i(R_j(\theta),n_j)|\mathbf{n_i},\mathbf{h_i},s_i,x_i=j]=\sum_{w \in \Theta} {g_{i,w} E[u_i(R_j(w),n_j)|\mathbf{n_i},\mathbf{h_i},s_i,x_i=j,\theta=w]}.
\end{equation}

Note that the decisions of customers $i+1,...,N$ are unknown to customer $i$ when customer $i$ makes the decision. Therefore, a close-form solution to (\ref{eqn35}) is generally impossible and impractical. In this paper, we purpose a recursive approach to compute the expected utility.

\subsection{Recursive Form of Best Response}\label{sec_gen_recur}
Let $BE_{i+1}(\mathbf{n_{i+1}},h_{i+1},s_{i+1})$ be the best response function of customer $i+1$. Then, according to $BE_{i+1}(\mathbf{n_{i+1}},\mathbf{\mathbf{h_{i+1}}},s_{i+1})$, the signal space $S$ can be partitioned into $S_{i+1,1},...,S_{i+1,K}$ subspaces with
\begin{equation}\label{eqn36}
    S_{i+1,j}(\mathbf{n_{i+1}},\mathbf{h_{i+1}}) = \{s | s \in S, BE_{i+1}(\mathbf{n_{i+1}},\mathbf{h_{i+1}},s) = j\},~~ \forall j\in\{1,...,K\}.
\end{equation}

Based on (\ref{eqn36}), we can see that, given $\mathbf{n_{i+1}}$ and $\mathbf{h_{i+1}}$, $BE_{i+1}(\mathbf{n_{i+1}},\mathbf{h_{i+1}},s_{i+1})=j$ if and only if $s_{i+1} \in S_{i+1,j}$. Therefore, the decision of customer $i+1$ can be predicted according to the signal distribution $f(s|\theta)$ given by
\begin{equation}\label{eqn37}
    Pr(x_{i+1}=j|\mathbf{n_{i+1}},\mathbf{h_{i+1}})= \int_{s \in S_{i+1,j}(\mathbf{n_{i+1}},\mathbf{h_{i+1}})} {f(s)ds}.
\end{equation}

Let us define $m_{i,j}$ as the number of customers choosing table $j$ after customer $i$ (including customer $i$ himself). Then, we have $n_j=n_{i,j}+m_{i,j}$, where $n_j$ denotes the final number of customers choosing table $j$ at the end of the game. Moreover, according to the definition of $m_{i,j}$, we have
 \begin{equation}\label{eqn38}
    m_{i,j}=\left\{
              \begin{array}{ll}
                1+m_{i+1,j}, & \hbox{$x_i=j$;} \\
                m_{i+1,j}, & \hbox{else.}
              \end{array}
            \right.
\end{equation}

The recursive relation of $m_{i,j}$ in (\ref{eqn38}) will be used in the following to get the recursive form of the best response function. We first derive the recursive form of the distribution of $m_{i,j}$, i.e., $Pr(m_{i,j}=X|\mathbf{n_{i}},\mathbf{h_{i}},s_{i},x_{i})$ can be expressed as a function of $Pr(m_{i+1,j}=X|\mathbf{n_{i+1}},\mathbf{h_{i+1}},s_{i+1},x_{i+1}=j,\theta=l),~\forall~l \in \Theta,~0 \leq j \leq K$, as follows:
\begin{eqnarray}\label{eqn39}
\!\!\!\!\!\!\!\!\!\!\!\!&&\!\!\!\!\!\!\!\!\!\!\!\!Pr(m_{i,j}=X|\mathbf{n_i},\mathbf{h_i},s_i,x_i,\theta=l) =\left\{
                                            \begin{array}{ll}
                                              Pr(m_{i+1,j}=X-1|\mathbf{n_i},\mathbf{h_i},s_i,x_i,\theta=l), & \hbox{$x_i=j$,} \\
                                              Pr(m_{i+1,j}=X|\mathbf{n_i},\mathbf{h_i},s_i,x_i,\theta=l), & \hbox{$x_i \neq j$,}
                                            \end{array}
                                          \right. \\
                                          \!\!\!\!\!\!\!\!\!\!\!\!&&\!\!\!\!\!\!\!\!\!\!\!\!=\!\!\!\left\{\!\!\!\!
                                               \begin{array}{ll}
                                                 \sum_{u \in \{1,...,K\}} \int_{s \in S_{i+1,u}(\mathbf{n_{i+1}},\mathbf{h_{i+1}})} \nonumber {Pr(m_{i+1,j}=X-1|\mathbf{n_{i+1}},\mathbf{h_{i+1}},s_{i+1}=s,x_{i+1}=u, \theta=l)} f(s|\theta=l) ds, &\!\!\! \hbox{$x_i=j$,} \\
                                                 \sum_{u \in \{1,...,K\}} \int_{s \in S_{i+1,u}(\mathbf{n_{i+1}},\mathbf{h_{i+1}})} {Pr(m_{i+1,j}=X|\mathbf{n_{i+1}},\mathbf{h_{i+1}},s_{i+1}=s,x_{i+1}=u, \theta=l)} f(s|\theta=l) ds, &\!\!\! \hbox{$x_i \neq j$,}
                                               \end{array}
                                             \right.
\end{eqnarray}
where $\mathbf{h_{i+1}}$ and $\mathbf{n_{i+1}}$ can be obtained using
\begin{equation}\label{eq_h_plus}
\mathbf{h_{i+1}}=\{h_{i},s_i\} \text{ and } \mathbf{n_{i+1}}=\{n_{i+1,1},...,n_{i+1,K}\},
\end{equation}
with
\begin{equation}\label{eq_n_plus}
n_{i+1,k}=\left\{
        \begin{array}{ll}
                 n_{i,k}+1, & \mbox{if $x_i=k$}, \\
                 n_{i,k}, & \mbox{otherwise}. \\
        \end{array}
        \right.
\end{equation}

Based on (\ref{eqn39}), $Pr(m_{i,j}=X|\mathbf{n_{i}},\mathbf{h_{i}},s_{i},x_{i},\theta=l)$ can be recursively calculated. Therefore, we can calculate the expected utility $E[u_i(R_j(\theta),n_j)|\mathbf{n_i}, \mathbf{h_i}, s_i]$ as
\begin{equation}\label{eq_exp_util_recur}
    E[u_i(R_j(\theta),n_j)|\mathbf{n_i}, \mathbf{h_i}, s_i] = \sum_{l \in \Theta} \sum_{x=0}^{N-i+1} g_{i,l} Pr(m_{i,j}=x|\mathbf{n_i},\mathbf{h_i},s_i,x_i=j,\theta=l)u_i(R_j(l),n_{i,j}+x).
\end{equation}
Finally, the best response function of customer $i$ can be derived by
\begin{equation}\label{eqn44}
    BE_i(\mathbf{n_i},\mathbf{h_i},s_i) = \arg\max_j \sum_{l \in \Theta} \sum_{x=0}^{N-i+1} g_{i,l}Pr(m_{i,j}=x|\mathbf{n_i},\mathbf{h_i},s_i,x_i=j,\theta=l)u_i(R_j(l),n_{i,j}+x).
\end{equation}

With the recursive form, the best response function of all customers can be obtained using backward induction. The best response function of the last customer $N$ can be found as
\begin{equation}\label{eqn45}
    BE_N(\mathbf{n_N},h_N,s_N) = \arg\max_j \sum_{l \in \Theta} g_{N,l} u_N(R_j(l),n_{N,j}+1).
\end{equation}
Note that $Pr(m_{N,j}=X|\mathbf{n_{N}},\mathbf{h_{N}},s_{N},x_{N},\theta)$ can be easily derived as follows:
\begin{equation}\label{eqn46}
    Pr(m_{N,j}=1|\mathbf{n_{N}},\mathbf{h_{N}},s_{N},x_{N},\theta)= \left\{
        \begin{array}{ll}
                 1, & \mbox{if $x_{N} = j$}, \\
                 0, & \mbox{otherwise}. \\
        \end{array}
        \right.
\end{equation}

\section{Simulation}\label{sec_sim}
In this section, we verify the proposed recursive best response and corresponding equilibrium. We simulate a Chinese restaurant with two tables $\{1,2\}$ and two possible states $\theta \in \{1,2\}$. When $\theta=1$, the size of table $1$ is $R_1(1)=100$ and the size of table $2$ is $R_2(1)=40$. When $\theta=2$, $R_1(2)=40$ and $R_2(2)=100$. The state is uniformly randomly chosen at the beginning of each simulation with probability $0.5$. The number of customers is fixed. Each customer receives a randomly generated signal $s_i$ at the beginning of the simulation. While conditioning on the system state, the signals customers received are independent. The signal distribution $f(s|\theta)$ is given by
\begin{equation}
    Pr(s=1|\theta=1) = Pr(s=2|\theta=2) = p~,~Pr(s=2|\theta=1) = Pr(s=1|\theta=2) = 1-p,
\end{equation}
where $p \geq 0.5$ can be regarded as the quality of signals. When the signal quality $p$ is closer to $1$, the signal is more likely to reflect the true state $\theta$. With the signals, customers make their decisions sequentially. After the $i$-th customer makes his choice, he reveals his decision and signal to other customers. The game ends after the last customer made his decision. Then, the utility of the customer $i$ choosing table $j$ is given by $U_i(R_j(\theta),n_j)=\frac{R_j}{n_j}$, where $n_j$ is the number of customers choosing table $j$ at the end of the game.

%\subsection{Verification}
%We first validate the theoretical recursive best response function by checking whether the expected utility of certain actions shown in (\ref{eq_exp_util_recur}) fits the simulation results. We simulate with the case that the number of customers is $N=5$. We assume that $p \in [0.5,1]$. For each $p$, the game is simulated with $100,000$ times. The average utility of the simulations under the same signal quality are shown in Fig. \ref{fig_sim_ver}. The expected utility, which is computed by (\ref{eq_exp_util_recur}), is also shown in the same figure. We can see that the simulated average utilities match with the expected ones, which validates the correctness of (\ref{eq_exp_util_recur}).
% means that (\ref{eq_exp_util_recur}) correctly estimates the expected utility of a customer after choosing certain actions.

%\begin{figure}
%   \begin{centering}
%       \includegraphics[width=9cm]{sim_5_40_ver.eps}
%       \caption{Verification: Expected Utility vs. Simulated Average Utility}
%   \label{fig_sim_ver}
%   \end{centering}
%\end{figure}

\subsection{Optimality of Proposed Best Response}
We first verify the optimality of the best response. We study the $5$-customer scheme with the same settings in the previous simulation. We assume that all customers except customer $2$ apply their best response strategies while customer $2$ chooses the table opposite to his best response strategy with a miss probability $p^{mis}$. We assume the signal quality $p \in [0.5,1]$ and $p^{mis} \in [0,1]$ in the simulations.

\begin{figure}
    \begin{centering}
        \subfigure[Customer 2]{
        \includegraphics[width=7cm]{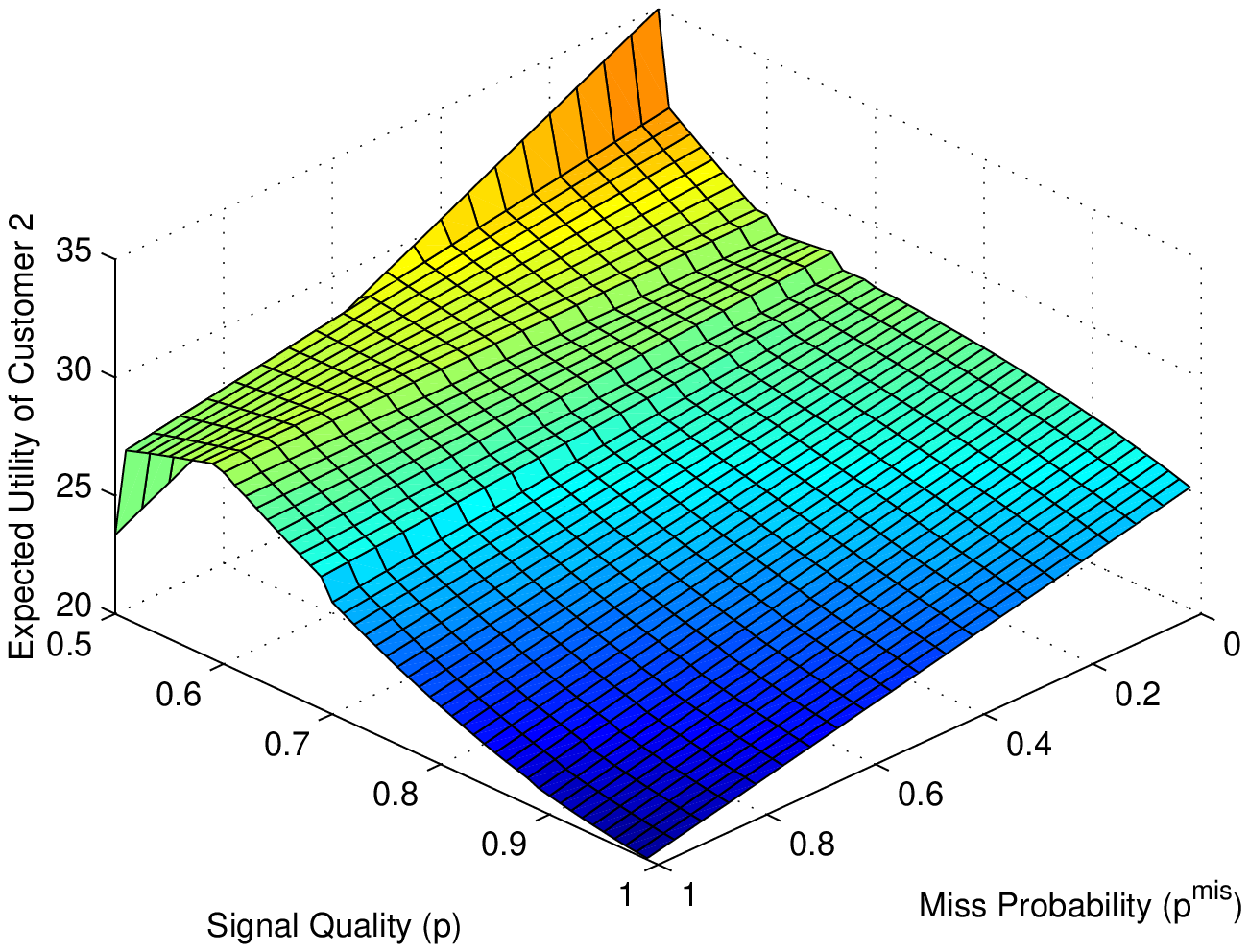}
        \label{fig_sim_check_2}
      }
      \subfigure[Expected Utility Increase of Customer 3]{
            \includegraphics[width=7cm]{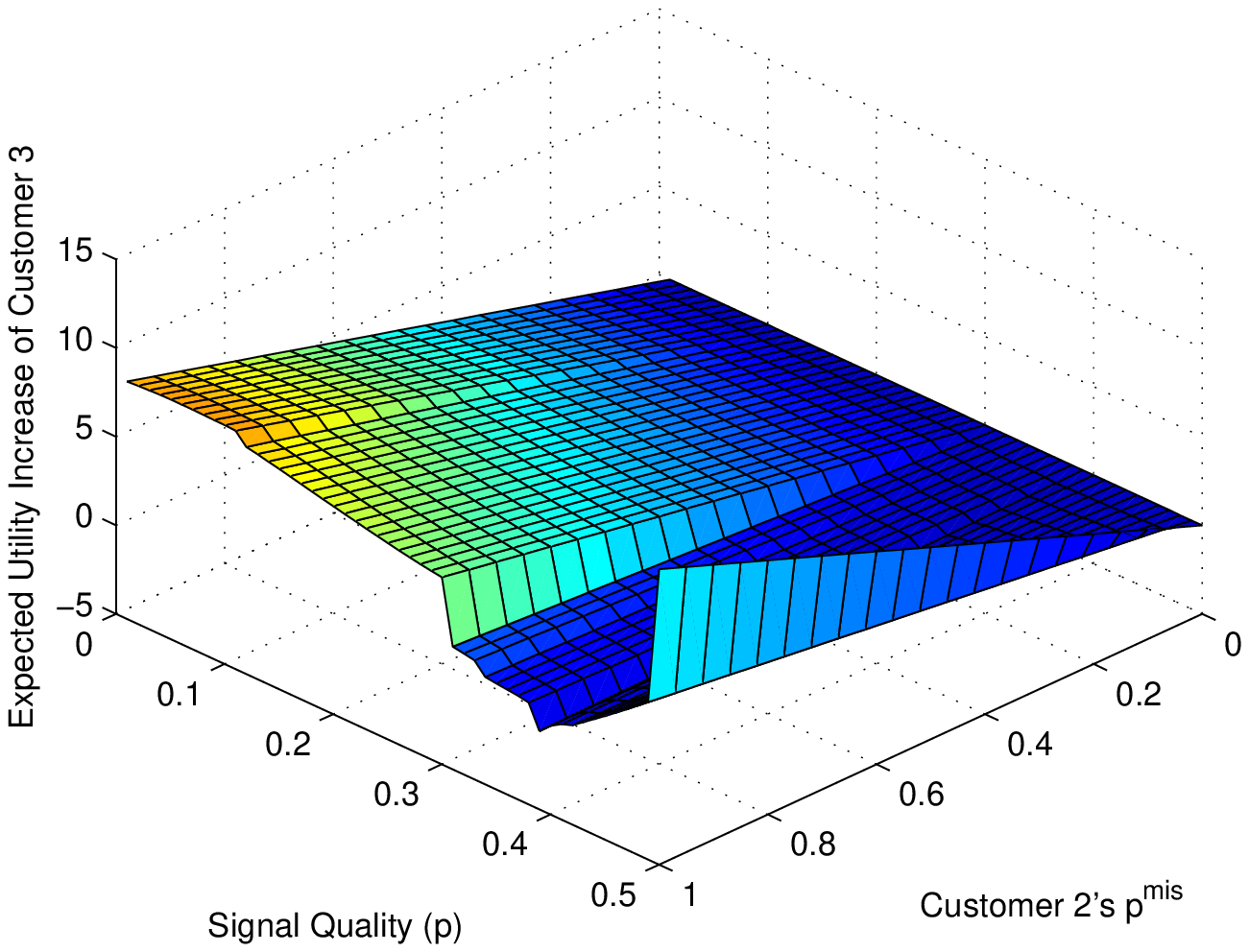}
        \label{fig_sim_check_3}
      }
    \caption{Expected utility of Customers under different Miss Probability of Customer $2$}
    \label{fig_sim_check}
    \end{centering}

\end{figure}

From Fig. \ref{fig_sim_check_2}, we can see that when $p^{mis}$ increases, the expected utility of customer $2$ always decreases for any signal quality $p$. This confirms the optimality of the proposed best response function. We also observe that when customer $2$ has a positive $p^{mis}$, at least one of the other customers will have a better average utility. Using the expected utility increase of customer $3$ shown in Fig. \ref{fig_sim_check_3} as an example, when customer $2$ makes a mistake in choosing the table, customer $3$ benefits from the mistake of customer $2$ under most signal qualities.

\subsection{Advantage of Playing Positions vs. Signal Quality}
Next, we investigate how the decision order and quality of signals affect the utility of customers. We follow the same settings in previous simulations except the table sizes. We fix the size of one table as $100$. The size of the other table is $r\times100$, where $r$ is the ratio of the table sizes. In the simulations, we assume the ratio $r \in [0,1]$. When the ratio $r=1$, two tables are identical, but the utility of choosing each table may have different utility since we may have odd customers. When $r=0$, one table has a size of $0$, which means a customer has a positive utility only when he chooses the correct table.

\begin{figure}
    \begin{centering}

      \subfigure[\scriptsize{Customer with Largest Expected Utility}]{

            \includegraphics[width=5cm]{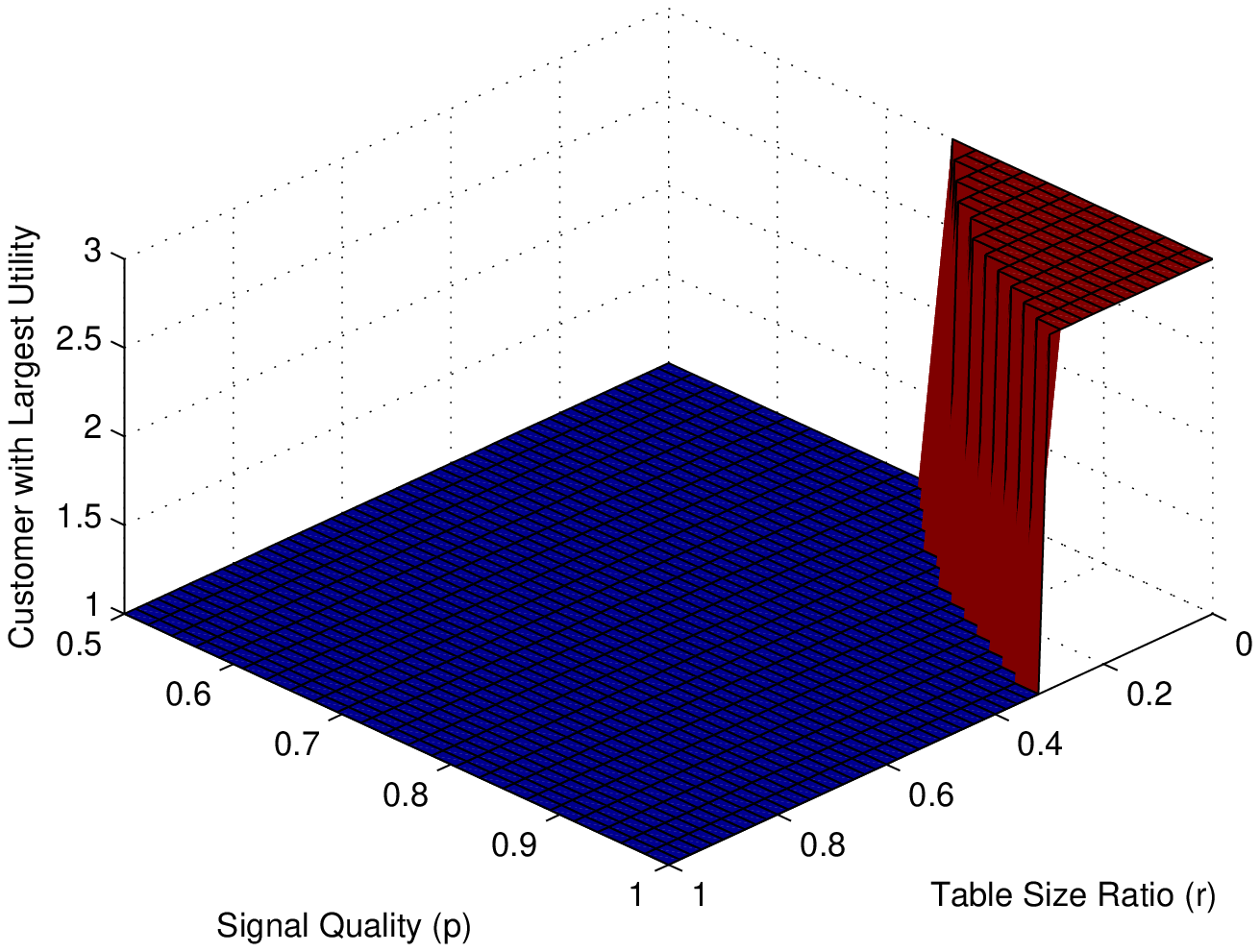}
        \label{fig_sim_adv_3_max}
      }
      \subfigure[\scriptsize{Utility Difference between $1$ and $2$}]{
        \includegraphics[width=5cm]{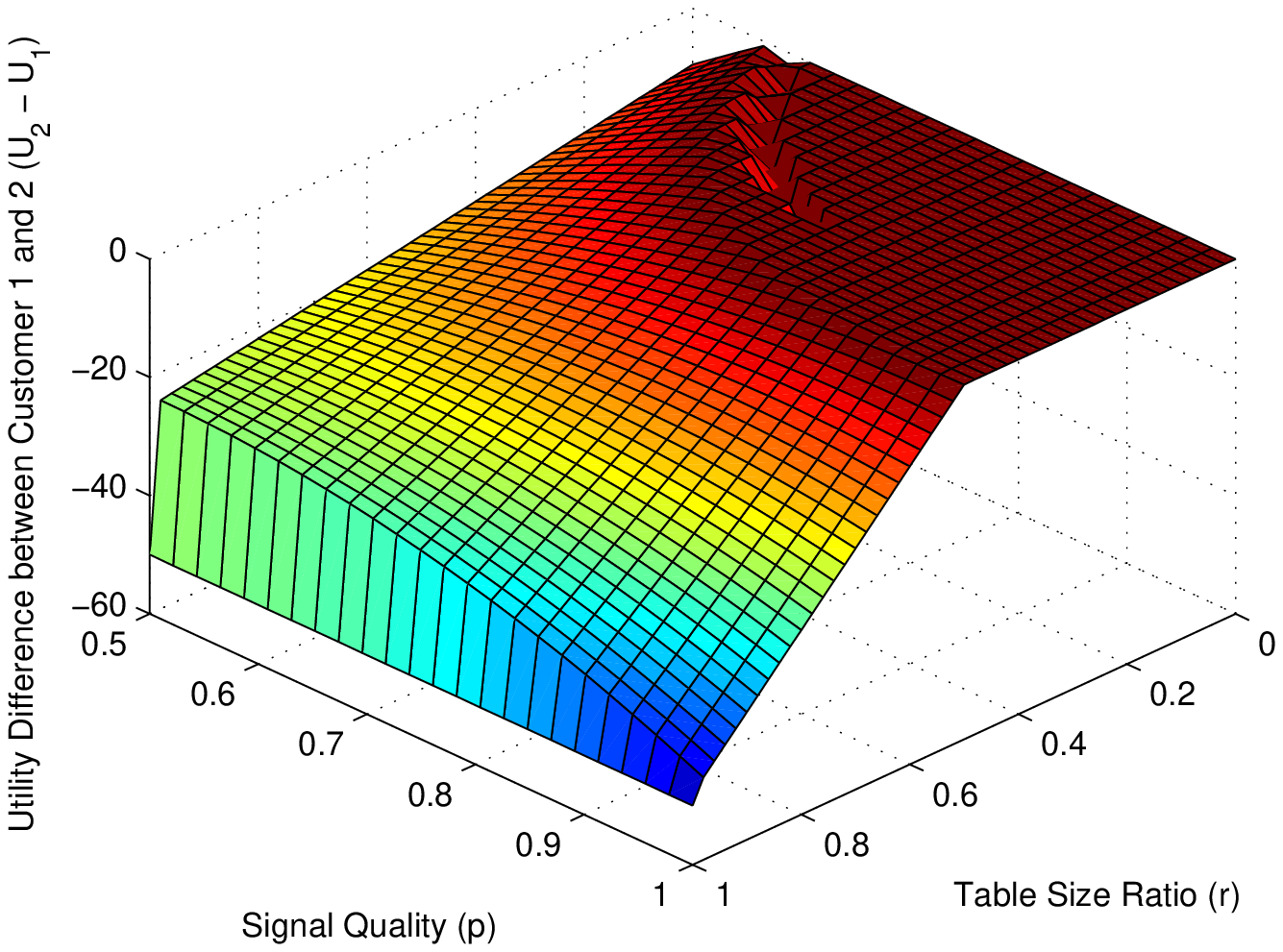}
        \label{fig_sim_adv_3_2}
      }
      \subfigure[\scriptsize{Utility Difference between $1$ and $3$}]{
        \includegraphics[width=5cm]{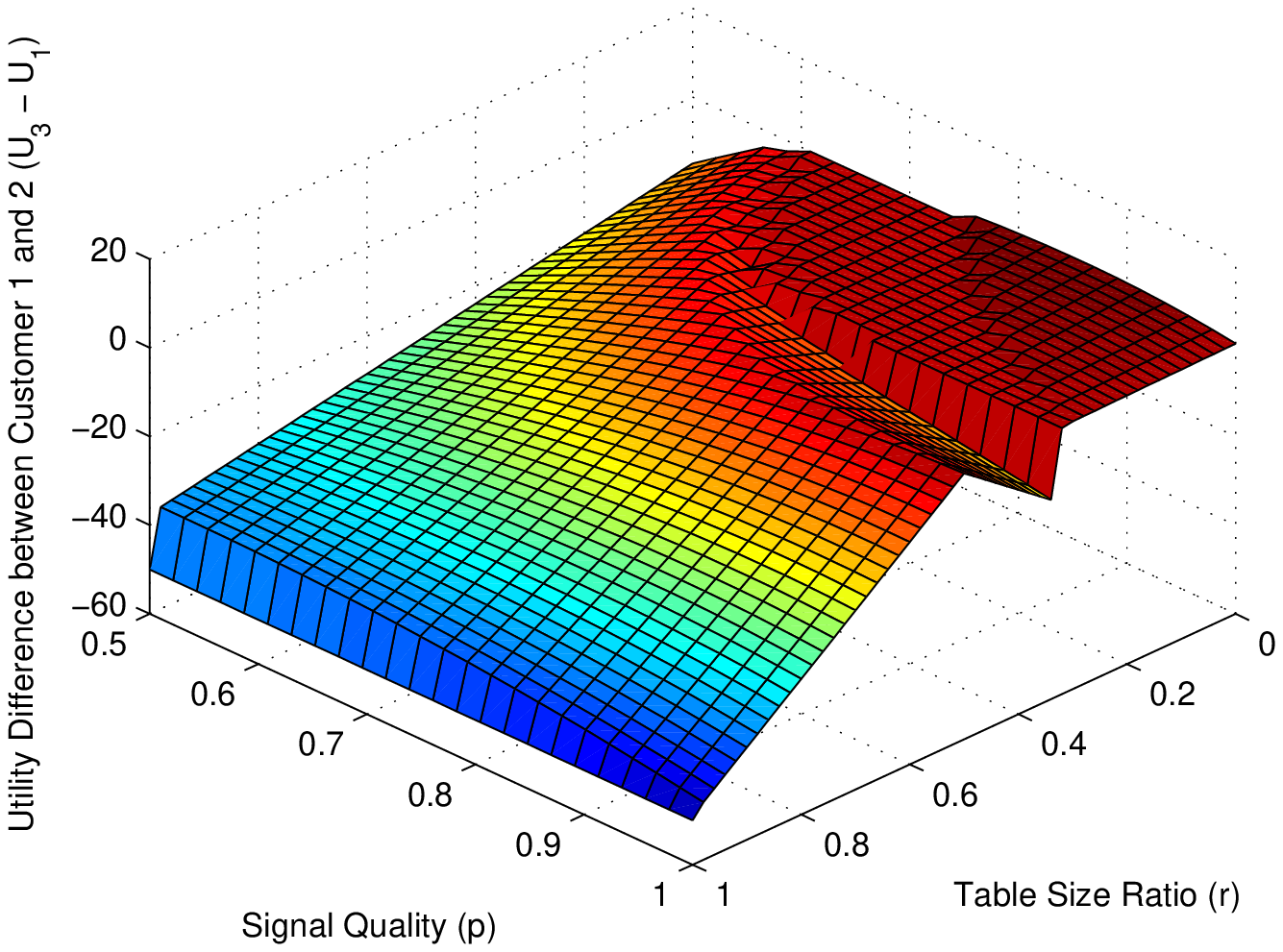}
        \label{fig_sim_adv_3_3}
      }
    \caption{The effect of different Table Size Ratio and Signal Quality in 3-Customer Restaurant}
    \label{fig_sim_adv_3}
    \end{centering}
\end{figure}

We first simulate a 3-customer scheme. From Fig. \ref{fig_sim_adv_3}, we can see that the advantage of customers making decisions at different order is significantly affected by both the signal quality and the table size ratio. As shown in Fig. \ref{fig_sim_adv_3_max}, when the signal quality is high and the table size ratio is low, customer $3$ has the largest expected utility. For other regions, customer $1$ has the largest expected utility. This phenomenon can be explained as follows. When the ratio is lower than $\frac{1}{3}$, all customers desire the larger table since even all of them select the larger one, each of them still have a utility larger than choosing the smaller one. In such a case, customers who choose late would have advantages since they have collected more signals and have a higher probability to identify the large table.

On one hand, when the signal quality is low, even the third customer cannot form a strong belief on the true state. In such a case, the expected size of each table becomes less significantly, and customers' decisions rely more on the negative network externality effect, i.e., how crowded of each table. When the first two customers choose the same table, customer $3$ is more likely to choose the other table to avoid the negative network externality. On the other hand, when the signal quality is high, customer $3$ is likely to form a strong belief on the true state and will choose the table according to the signals he collected. Therefore, when signal quality is high and the table size ratio is low, customer $3$ has the advantage of getting a higher utility in the Chinese restaurant game.

\begin{figure}

    \begin{centering}
      \subfigure[5 Customers]{
            \includegraphics[width=7cm]{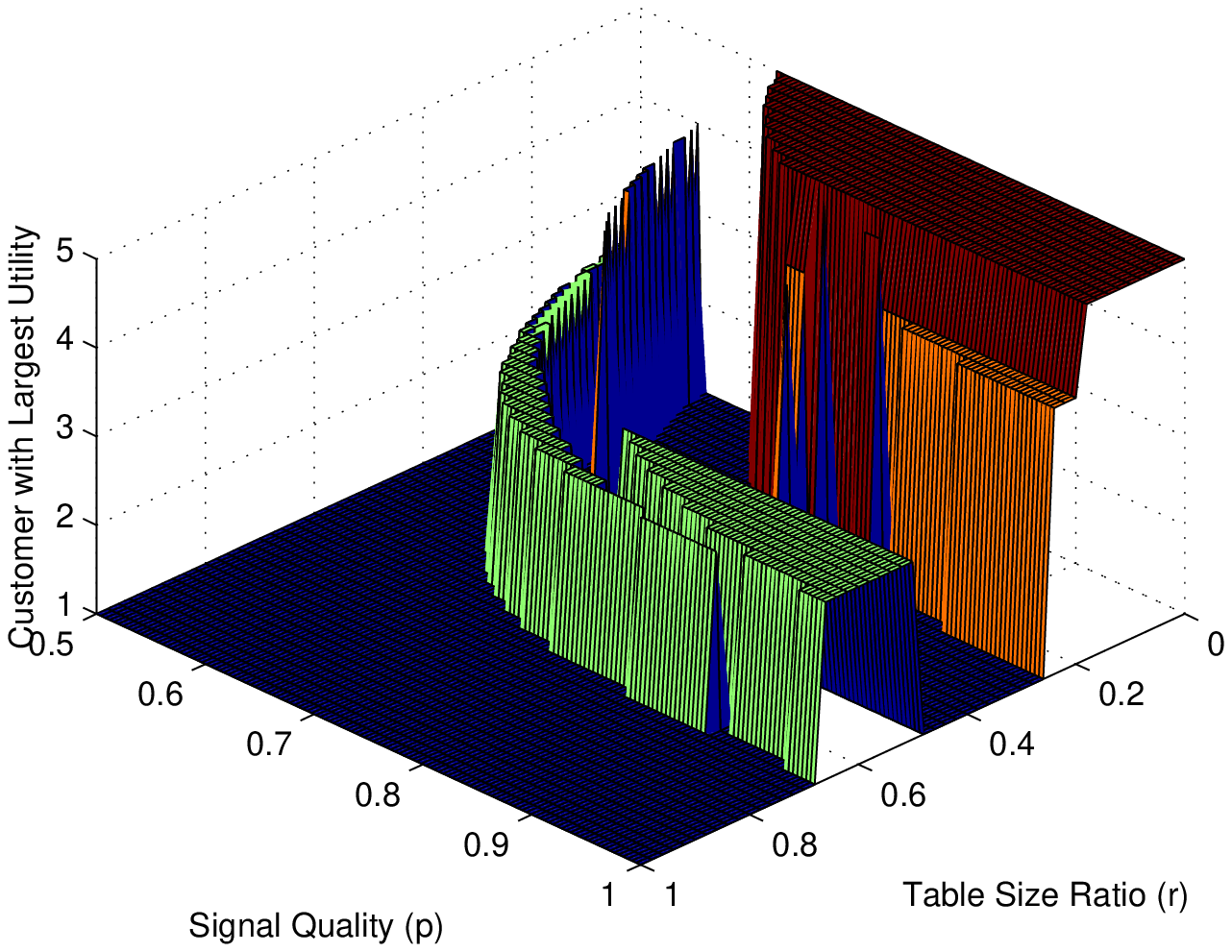}
        \label{fig_sim_adv_5}
      }
        \subfigure[10 Customers]{
            \includegraphics[width=7cm]{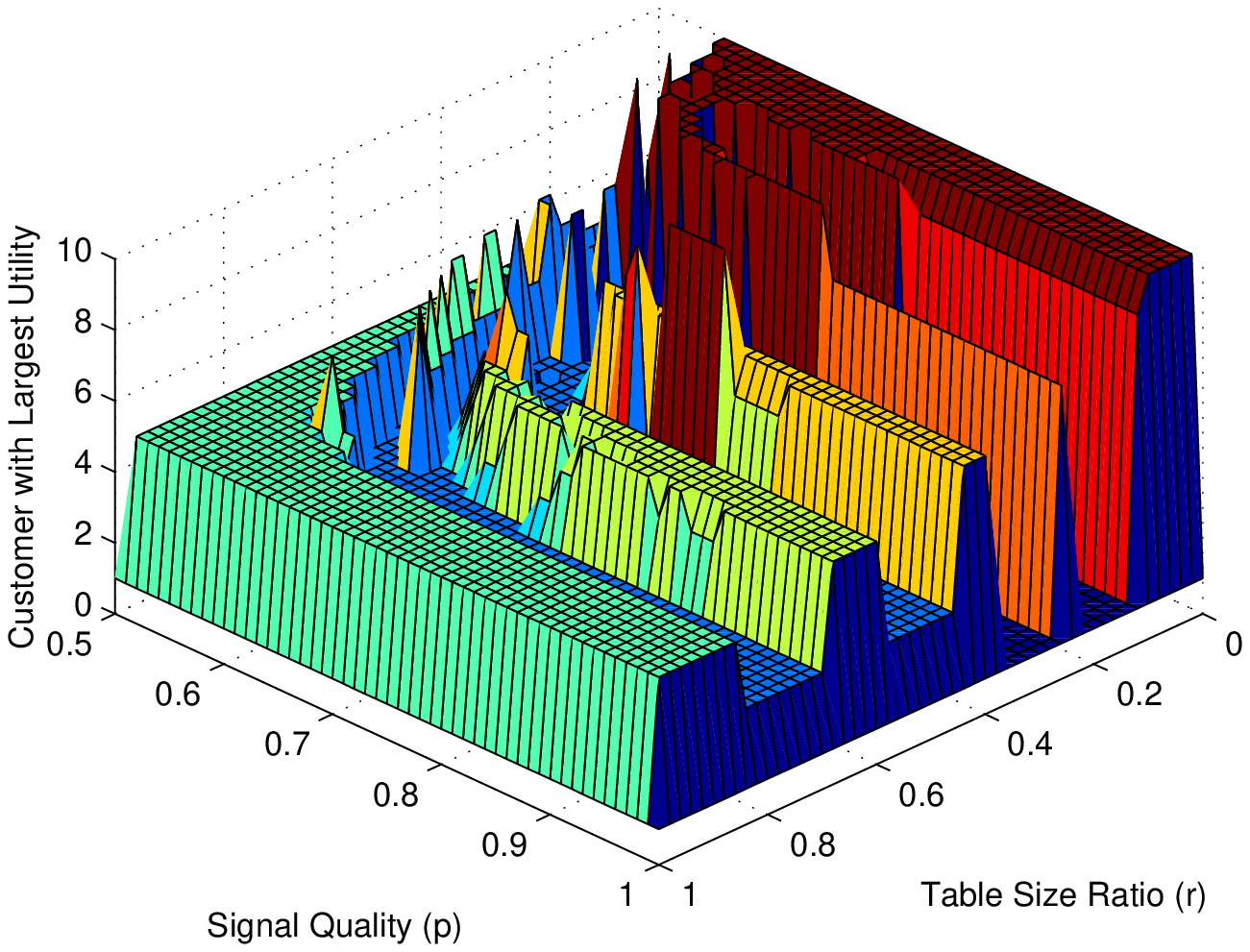}
        \label{fig_sim_adv_10}
      }
  \caption{The effect of different Table Size Ratio and Signal Quality in 5- and 10-Customer Restaurants}
    \label{fig_sim_adv_more}
    \end{centering}
\end{figure}

Nevertheless, due to the complicated game structure in Chinese restaurant game, the effect of signal quality and table size ratio is generally non-linear. As shown in Fig. \ref{fig_sim_adv_5}, when the number of customers increases to $5$, similar to the 3-customer scheme, customer $5$ has the largest utility when the signal quality is high and the table size ratio is low, while customer $1$ has the largest utility when the signal quality is low and the table size ratio is high. However, we observe that in some cases, customer $3$ becomes the one with largest utility. The reasons behind this phenomenon is as follows. In these cases, we observe that the expected number of customers in the larger table is $3$, and this table provides the customers a larger utility then the other one at the equilibrium. Therefore, customers would try to identify this table and choose it according to their own believes. Since customer $3$ collects more signals than customer $1$ and $2$, he is more likely to choose the correct table. Moreover, since he is the third customer to choose a table, this table is always available to him. Therefore, customer $3$ has the largest expected utility in these cases.

Note that the expected table size is determined by both the signal quality and the table size ratio. Generally, when the signal quality is low, a customer is less likely to construct a strong belief on the true state, i.e., the expected table sizes of both tables are similar. This suggests that a lower signal quality has a similar effect on the expected table size as a higher table size ratio. Our arguments are supported by the concentric-like structure shown in Fig. \ref{fig_sim_adv_5}. The same arguments can be applied to the 10-customer scheme, which is shown in Fig. \ref{fig_sim_adv_10}. We can observe the similar concentric-like structure. Additionally, we observe that when the table size ratio increases, the order of customer who has the largest utility in the peaks decreases from $10$ to $5$. This is consistent with our arguments since when the table size ratio increases, the equilibrium number of customers in the large table decreases from $10$ to $5$. This also explains why customer $1$ does not have the largest utility when the table size ratio is high. In this case, the equilibrium number of customers in the large table is $5$, and the large table provides higher utilities to customers in the equilibrium.
Since customer $5$ can collect more signals than previous customers, he has better knowledge on the table size than customer 1 to 4. Moreover, since customer $5$ is the fifth one to choose the table, he always has the opportunity to choose the large table. In such a case, customer 5 is the one with the largest expected utility when the table size ratio is high.

Next, we discuss two specific scenarios: the resource pool scenario with $r=0.4$ and available/unavailable scenarios with $r=0$. In resource pool scenario, the table size of the second table is $40$. Users act sequentially and rationally to choose these two tables to maximize their utilities. In available/unavailable scenario, the second table size is $0$, which means that a customer has positive utility only when he chooses the right table. For both scenarios, we examine the schemes with the number of customers $N=3$ and $N=5$.

\begin{figure}

    \begin{centering}
      \subfigure[5 Customers]{
            \includegraphics[width=5.5cm]{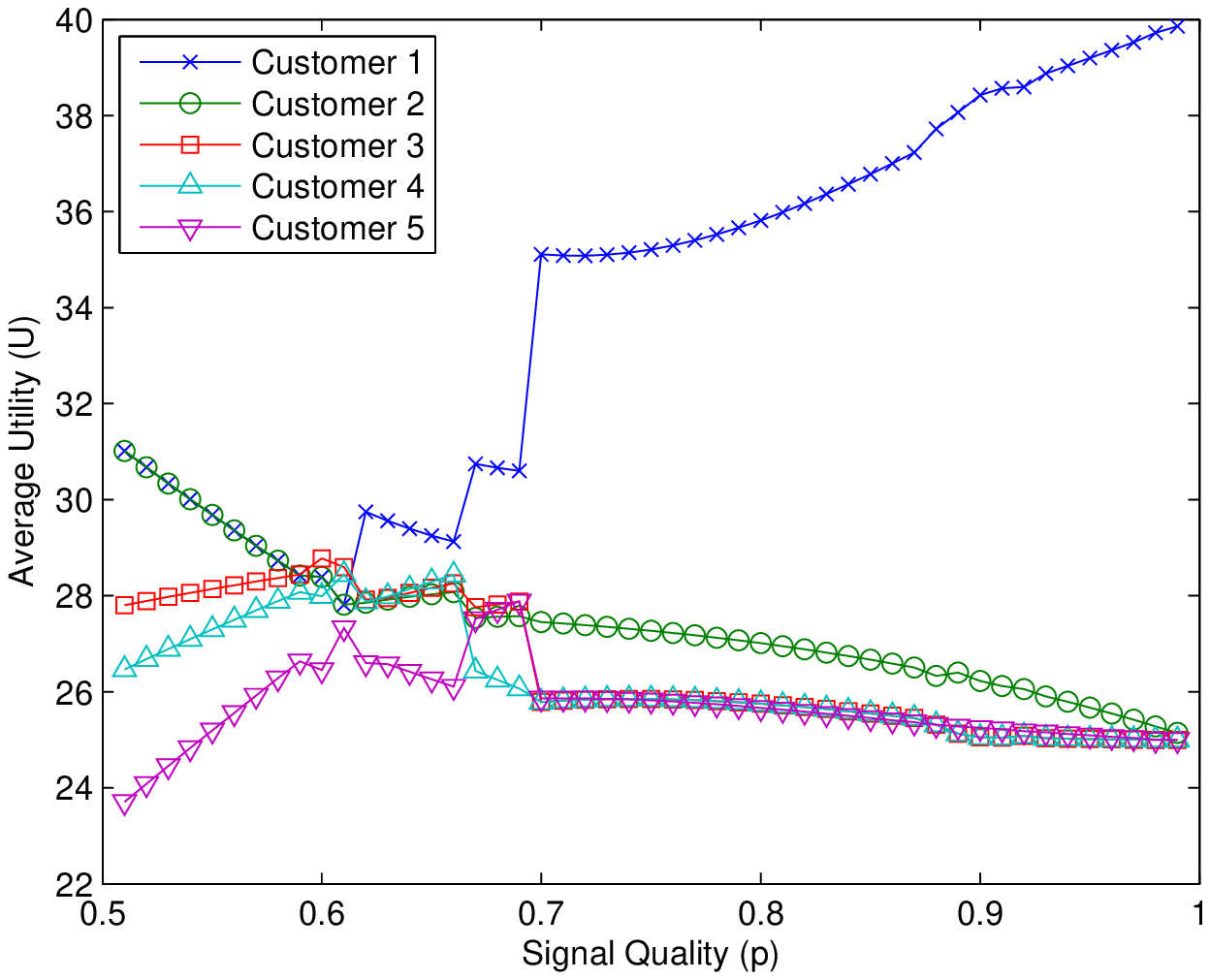}
        \label{fig_sim_pool_5}
      }
      \subfigure[3 Customers]{
        \includegraphics[width=5.5cm]{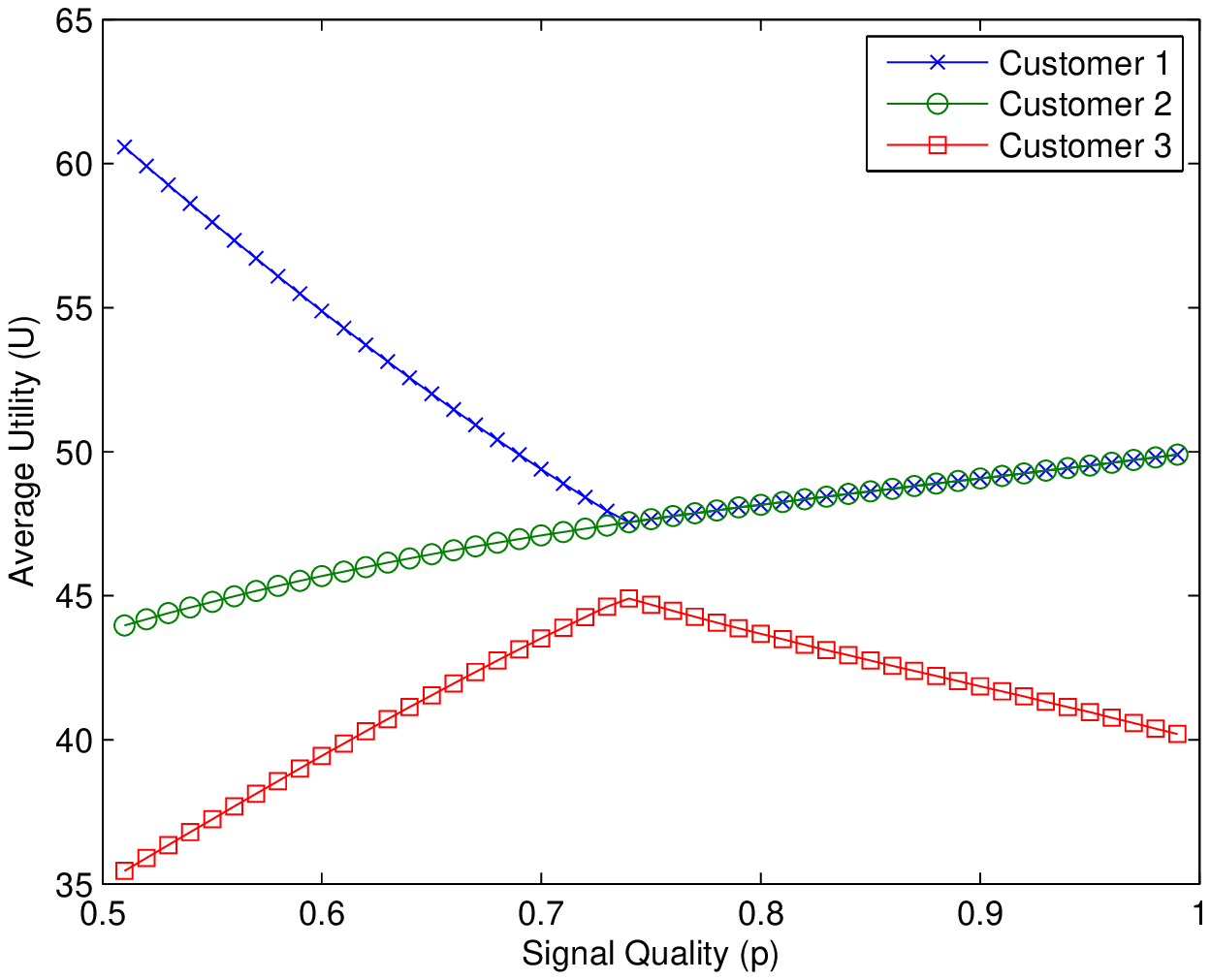}
        \label{fig_sim_pool_3}
      }
      \subfigure[Best Response when N=3]{
      \scriptsize
      \begin{tabular}[b]{|c|c|c|}
        \hline
        Signals & \multicolumn{2}{c|}{ Best Response } \\
        \hline
        $s_1,s_2,s_3$ & $p=0.9$ & $p=0.6$ \\
        \hline
         2,2,2 &  2,2,1 & 1,2,2 \\
     1,2,2 &  1,2,2 & 2,1,2 \\
     2,1,2 &  2,1,2 & 1,2,2 \\
     1,1,2 &  1,1,2 & 2,1,1 \\
     2,2,1 &  2,2,1 & 1,2,2 \\
     1,2,1 &  1,2,1 & 2,1,1 \\
     2,1,1 &  2,1,1 & 1,2,1 \\
     1,1,1 &  1,1,2 & 2,1,1 \\
     \hline
        \end{tabular}
        \label{tab_be_pool_3}
      }
    \caption{Average utility of Customers in Resource Pool scenario when $r=0.4$}
    \label{fig_sim_pool}
    \end{centering}
\end{figure}

From Fig. \ref{fig_sim_pool}, we can see that in the resource pool scenario with $r=0.4$, customer $1$ on average has significant higher utility, which is consistent with the result in Fig. \ref{fig_sim_adv_5}. Using 5-customer scheme shown in Fig. \ref{fig_sim_pool_5} as an example, the advantage of playing first becomes significant when signal quality is very low ($p < 0.6$), or the signal quality is high ($p > 0.7$).  We also find that customer $5$ has the lowest average utility for most signal quality $p$. We may have a clearer view on this in the 3-customer scheme. We list the best response of customers given the received signals in Fig. \ref{tab_be_pool_3}. We observe that when signal quality $p$ is large, both customer $1$ and $2$ follow the signals they received to choose the tables. However, customer $3$ does not follows his signal if the first two customers choose the same table. Instead, customer $3$ will choose the table that is still empty. In this case, although customer $3$ know which table is larger, he does not choose that table since it has been occupied by the first two customers. The network externality effect dominates the learning advantage in this case.

However, when $p$ is low, the best response of customer $1$ is opposite, i.e., he will choose the table that is indicated as the smaller one by the signal he received. At the first glance, the best response of customer $1$ seems to be unreasonable. However, such a strategy is indeed customer $1$'s best response considering the expected equilibrium in this case. According to Theorem \ref{thm_exist_ne_seq_p}, if perfect signals ($p=1$) are given, the large table should be chosen by customer $1$ and $2$ since the utility of large table is $100/2=50$ is larger than the that of the small table, which is $40/1=40$, in the equilibrium. However, when the imperfect signals are given, customers choose the tables based on the expected table sizes. When signal quality is low, the uncertainty on the table size is large, which leads to similar expected table sizes for both tables. In such a case, customer $1$ favors the smaller table because it can provide a higher expected utility, compared with sharing with another customer in the larger table.

\begin{figure}

    \begin{centering}
      \subfigure[5 Customers]{
            \includegraphics[width=5cm]{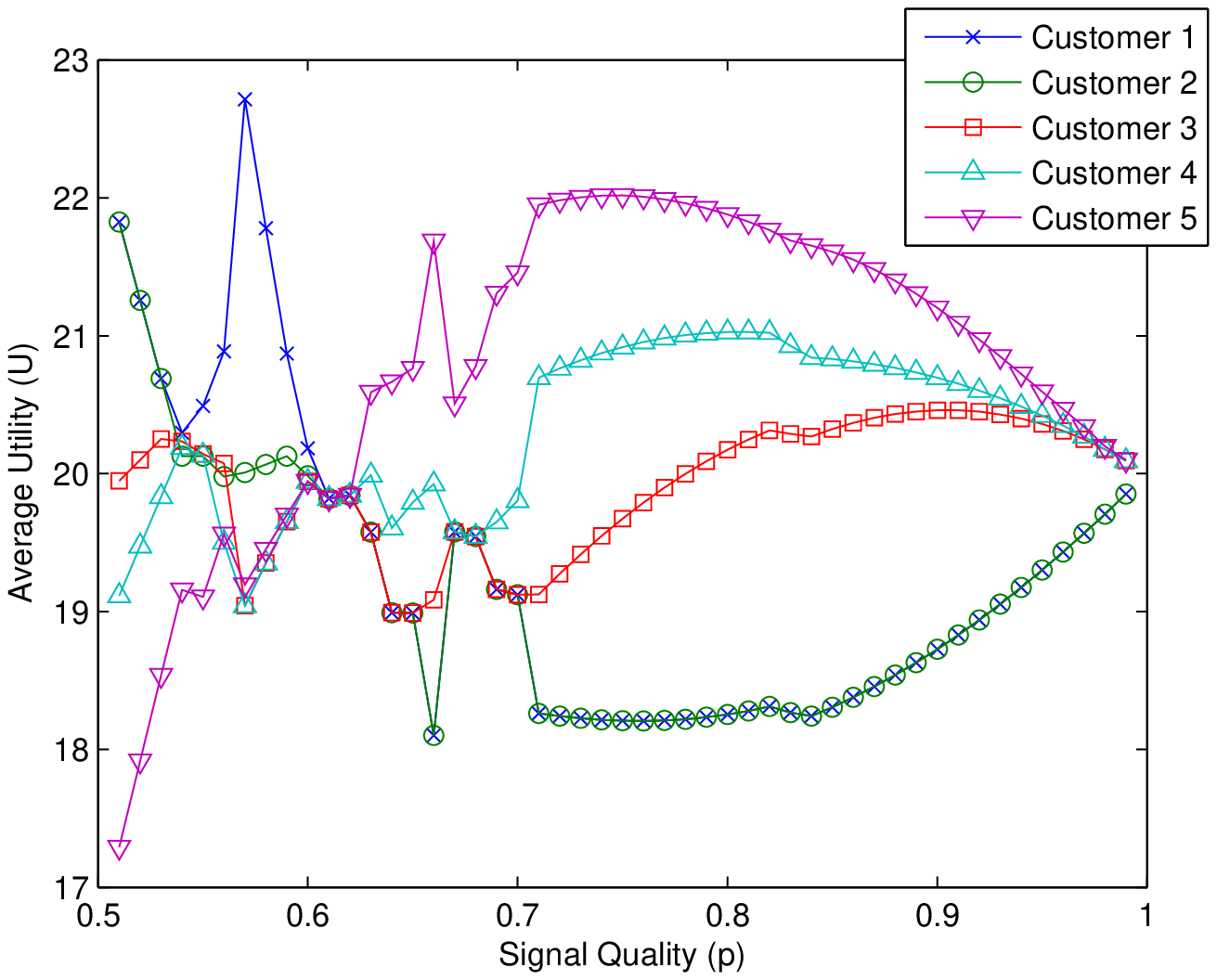}
        \label{fig_sim_avail_5}
      }
      \subfigure[3 Customers]{
        \includegraphics[width=5cm]{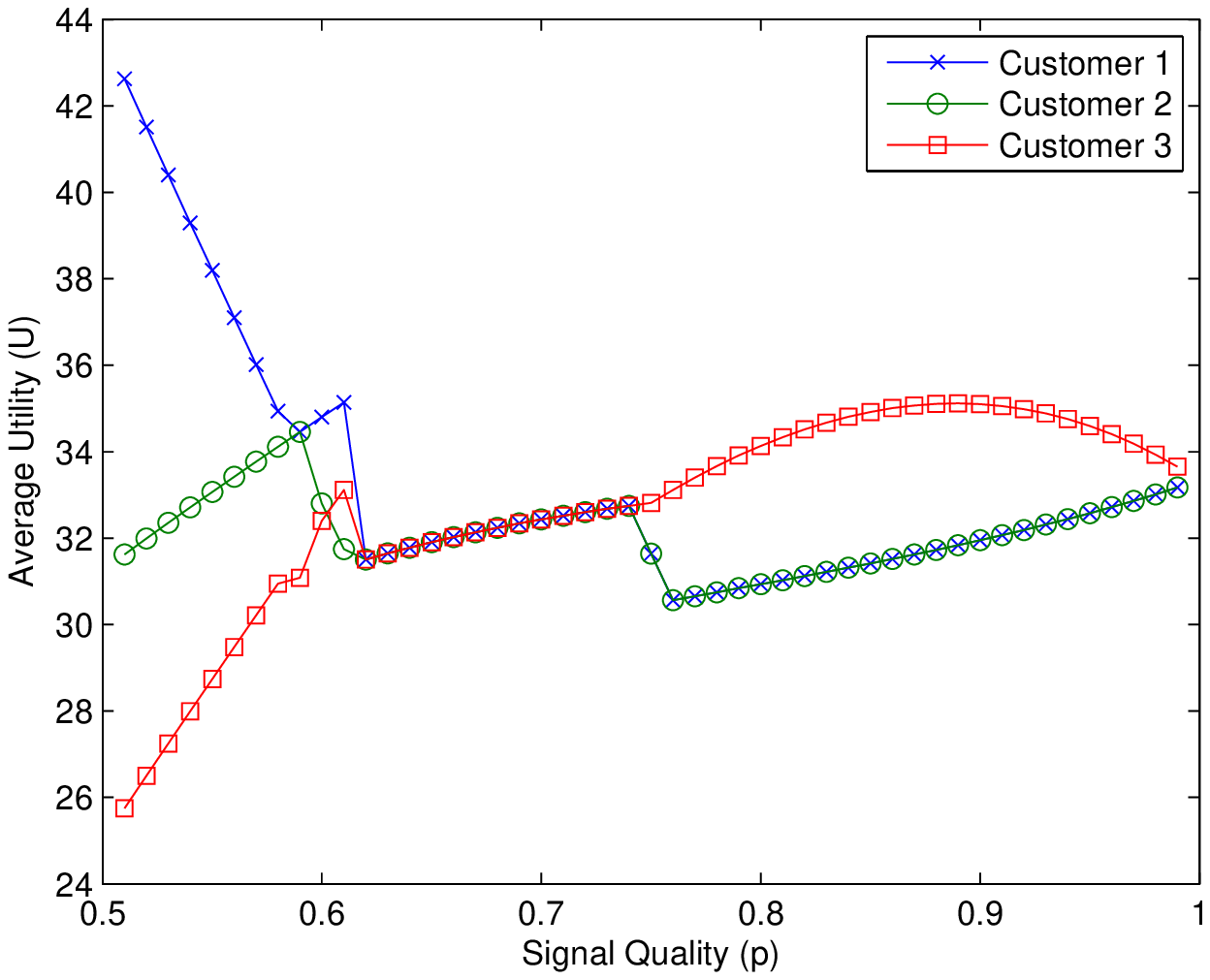}
        \label{fig_sim_avail_3}
      }
      \subfigure[Best Response when $N=3$]{
      \tiny
      \begin{tabular}[b]{|c|c|c|c|}
        \hline
        Signals & \multicolumn{3}{c|}{ Best Response }\\
        \hline
        $s_1,s_2,s_3$ & $p=0.9$ & $p=0.7$ & $p=0.55$ \\
        \hline
         2,2,2 & 2,2,2 & 2,2,2 & 1,2,2 \\
     1,2,2 & 1,2,2 & 1,2,2 & 2,1,2 \\
     2,1,2 & 2,1,2 & 2,1,2 & 1,2,2 \\
     1,1,2 & 1,1,1 & 1,1,2 & 2,1,1 \\
     2,2,1 & 2,2,2 & 2,2,1 & 1,2,2 \\
     1,2,1 & 1,2,1 & 1,2,1 & 2,1,1 \\
     2,1,1 & 2,1,1 & 2,1,1 & 1,2,1 \\
     1,1,1 & 1,1,1 & 1,1,1 & 2,1,1 \\
     \hline
        \end{tabular}
        \label{tab_be_avail_3}
        }
    \caption{Average utility of Customers in Available/Unavailable scenario when $r=0$}
    \label{fig_sim_avail}
    \end{centering}
\end{figure}

In the available/unavailable scenario, as shown in Fig. \ref{fig_sim_avail}, the advantage of customer $1$ in playing first becomes less significant. Using 5-customer scheme shown in Fig. \ref{fig_sim_avail_5} as an example, when signal quality $p$ is larger than $0.6$, customer $5$ has the largest average utility and customer $1$ has smallest average utility. Such a phenomenon is because customers should try their best on identifying the available table when $r=0$. Learning from previous signals gives the later customers a significant advantage in this case. Nevertheless, we observe that the best responses of later customers are not necessary always choosing the table that is more likely to be available. We use the 3-customer as an illustrative example. We list the best response of all customers given the received signals in Fig. \ref{tab_be_avail_3}. When the signal quality is pretty low ($p = 0.55$), we have the same best response as the one in resource pool scenario, where the network externality effect still plays a significant role. Using $(s_1,s_2,s_3)=(2,2,1)$ as an example, even customer $3$ finds that table $2$ is more likely to be available, his best response is still choosing table $1$ since table $2$ is already chosen by both customer $1$ and $2$, and the expected utility of choosing table $1$ with only himself is higher than that of choosing table $2$ with other two customers.
As the signal quality $p$ becomes high, e.g., $p = 0.9$, customer $3$ will choose the table according to all signals $s_1,s_2,s_3$ he collected. The belief constructed by the signals are now strong enough to overcome the loss in the network externality effect, which makes him choose the table that is more likely to be available.

\section{Conclusion}\label{sec_con}
In this paper, we proposed a new game, called Chinese Restaurant Game, by combining the strategic game-theoretic analysis and non-strategic machine learning technique. The proposed Chinese restaurant game can provide a new general framework for analyzing the strategic learning and predicting behaviors of rational agents in a social network with negative network externality.
By conducting the analysis on the proposed game, we derived the optimal strategy for each agent and provided a recursive method to achieve the equilibrium.
The tradeoff between two contradictory advantages, which are making decisions earlier for choosing better tables and making decisions later for learning more accurate believes, is discussed through simulations.
We found that both the signal quality of the unknown system state and the table size ratio affect the expected utilities of customers with different decision orders. Generally, when the signal quality is low and the table size ratio is high, the advantage of playing first dominates the benefit from learning. On the contrary, when the signal quality is high and the table size ratio is low, the advantage of playing later for better knowledge on the true state increases the expected utility of later agents.

\scriptsize

\bibliography{crg}
\bibliographystyle{unsrt}

\end{document}